\documentclass[journal]{IEEEtran}
\IEEEoverridecommandlockouts         

\usepackage{hyperref}
\usepackage{graphicx} 
\usepackage{times} 
\usepackage{amsmath} 
\usepackage{amssymb}  

\let\theoremstyle\relax
\usepackage{amsthm}
\usepackage{pgfplots}
\usepackage{subfigure}
\pgfplotsset{every axis plot post/.append style={rounded corners=.5\pgflinewidth}}

\makeatletter \let\cl@part\relax \makeatother

\usepackage{cleveref}

\usepackage{array}
\usepackage{accents}

\crefname{figure}{Fig.}{Fig.}
\crefname{table}{Tab.}{Tab.}
\crefname{equation}{}{}
\crefname{section}{Section}{Sections}
\crefname{theorem}{Theorem}{Theorems}
\crefname{lemma}{Lemma}{Lemmas}
\crefname{defin}{Definition}{Definitions}

\theoremstyle{remark}
\newtheorem{theorem}{Theorem}
\newtheorem{lemma}{Lemma}
\newtheorem{ex}{Example}
\newtheorem*{remark}{Remark}
\newtheorem{defin}{Definition}

\newcommand{\abs}[1]{\left|#1\right|}
\newcommand{\field}[1]{\mathbb{#1}}
\newcommand{\R}{\field{R}}
\newcommand{\RH}{\field{RH}_\infty}
   
\newcommand{\RL}{\field{RL}_\infty}
\newcommand{\C}{\field{C}}

\newcommand{\bmtx}{\begin{bmatrix}}
\newcommand{\emtx}{\end{bmatrix}}
\newcommand{\bsmtx}{\left[ \begin{smallmatrix}} 
\newcommand{\esmtx}{\end{smallmatrix} \right]} 
\newcommand{\bmatarray}[1]{\left[\begin{array}{#1}}
\newcommand{\ematarray}{\end{array}\right]} 

\newcommand{\Dt}{\mathcal{D}_\tau}
\newcommand{\St}{\mathcal{S}_\tau}
\newcommand{\Dtr}{\mathcal{D}_{\bar\tau,r}}
\newcommand{\Str}{\mathcal{S}_{\bar\tau,r}}
\newcommand{\Dtbar}{\mathcal{D}_{\bar \tau}}
\newcommand{\Stbar}{\mathcal{S}_{\bar \tau}}
\newcommand{\hatDt}{\hat{\mathcal{D}}_\tau}
\newcommand{\hatSt}{\hat{\mathcal{S}}_\tau}

\newcommand{\hatStbar}{\hat{\mathcal{S}}_{\bar \tau}}
\newcommand{\FT}{\mathcal{F}}

\begin{document}


\title{An Overview of Integral Quadratic Constraints for Delayed Nonlinear
and Parameter-Varying Systems}

\author{Harald Pfifer and Peter Seiler 
\thanks{H. Pfifer and P. Seiler are with the Aerospace Engineering and Mechanics Department, University of Minnesota, emails:
        {\tt\small hpfifer@umn.edu, seiler@aem.umn.edu}}%
} 

\date{}
\maketitle

\begin{abstract}

  A general framework is presented for analyzing the stability and
  performance of nonlinear and linear parameter varying (LPV) time
  delayed systems. First, the input/output behavior of the time delay
  operator is bounded in the frequency domain by integral quadratic
  constraints (IQCs). A constant delay is a linear, time-invariant
  system and this leads to a simple, intuitive interpretation for
  these frequency domain constraints. This simple interpretation is
  used to derive new IQCs for both constant and varying
  delays. Second, the performance of nonlinear and LPV delayed systems
  is bounded using dissipation inequalities that incorporate
  IQCs. This step makes use of recent results that show, under 
  mild technical conditions, that an IQC has an equivalent
  representation as a finite-horizon time-domain constraint. Numerical
  examples are provided to demonstrate the effectiveness of the
  method for both class of systems.


\end{abstract}


\section{Introduction}

This paper presents a framework to analyze nonlinear or linear
parameter varying (LPV) time-delayed systems.  In this framework the
system is separated into a nonlinear or LPV system in feedback with a
time delay. Stability and performance is considered for both constant
and varying delays.  The analysis uses the concept of integral
quadratic constraints (IQCs) \cite{Megretski1997}. Specifically, IQCs
describe the behavior of a system in the frequency domain in terms of
an integral constraint on the Fourier transforms of the input/output
signals.  Several IQCs valid for constant and varying delays have
already appeared in the literature, see
e.g. \cite{Megretski1997,Kao2004,Kao2007}.

This paper has two main contributions.  The first contribution is to
provide a simple interpretation for IQCs used to describe constant
time delays. In particular, constant time delays are linear,
time-invariant (LTI) systems and hence they have an equivalent
frequency response representation.  Thus IQCs valid for constant
delays can, in most cases, be interpreted as a frequency dependent
circle in the Nyquist plane.  It is noted that this interpretation
previously appeared in the robust control literature: ``G''-scales for
robustness analysis with real parameter uncertainty can be interpreted
with circles in the complex plane \cite{doyle85,fan91}.  Here, the
geometric interpretation is used to construct a new IQC valid for
constant delays. Moreover, the frequency domain intepretation provides
insight for generating a new IQC valid for time-varying delays even
though such varying delays are not LTI. All these results are
contained in Section~\ref{sec:freqdomain}.

The second contribution of this paper is to apply general IQCs for
analysis of nonlinear and LPV delayed systems.  The standard IQC
stability theorem in \cite{Megretski1997} was formulated with
frequency domain conditions and hence requires the ``nominal'' part of
the feedback interconnection to be an LTI system. An application of
this stability theorem to LTI systems in feedback with a constant
delay is given in \cite{fu97}. Previous work on delayed
nonlinear systems bounded the nonlinear elements of
the system and the time delays by IQCs and considered this
frequency domain approach to analyze a ``nominal'' LTI systems
under IQCs, see e.g. \cite{Peet2007}. Here, dissipation inequality conditions
are derived to assess the stability and performance of ``nominal'' nonlinear and
LPV systems in feedback with a delay. The dissipation inequalities are
time-domain conditions but IQCs are typically expressed as frequency
domain constraints.  Thus the key technical issue is that the analysis
approach requires an equivalent time-domain interpretation for an IQC.
Previous work along these lines for constant IQCs has appeared in
Chapter 8 of \cite{Gu2002}.  In fact, a large class of IQCs, under
mild technical conditions, have an equivalent expression as a
finite-horizon, time-domain integral as recently shown in
\cite{Megretski2010,seiler13}. This time-domain expression enables
IQCs to be easily incorporated into a dissipation inequality condition
as shown in Section~\ref{sec:timedomain}.  These analysis conditions
can be formulated and efficiently solved as sum-of-squares
optimizations \cite{parrilo00} and semidefinite programs \cite{befb94}
for nonlinear and LPV delayed systems, respectively. Numerical
examples for both system types are given in
Section~\ref{sec:examples}. These results complement recent robust
performance conditions for LPV systems
\cite{Scherer2012,Kose2009,Pfifer2014}. This paper builds upon \cite{pfifer14Aut}. In addition to the results in \cite{pfifer14Aut}, it includes a detailed description of the generation of IQC multipliers for time delays.

There is a large body of literature on time-delayed systems as
summarized in \cite{Gu2002, Briat2014}.  The most closely related work is that
contained in \cite{fu97,Megretski1997,Kao2004,Kao2007} which use IQCs
to derive stability conditions for LTI systems with constant or
varying delays. As noted above, the contribution of this paper is to
extend these results to nonlinear and LPV delayed systems. Lyapunov
theory is an alternative framework in the literature of time-delayed
systems \cite{Gu2002,Gu1997,Fridman2002}.  This approach essentially
constructs Lyapunov-Krasovskii or Lyapunov-Razumikhin functionals to
assess the convergence of the free (initial-condition) response of the
delayed system. Stability conditions for nonlinear
\cite{Papachristodoulou2004,Papachristodoulou2009} and LPV
\cite{Zhang2002} delayed systems have been developed in the Lyapunov
framework.  These methods treat the time delay as integrated with the
dynamics of the plant. This is in contrast to the approach considered
here which uses input-output stability (forced response) and separates
the time delay from the ``nominal'' plant dynamics. A benefit of the
IQC framework is that extends naturally to systems with many delays
and/or uncertainties.  On the other hand, some stability conditions in
the Lyapunov-Krasovskii framework appear to use time-varying quadratic
constraints that do not have counterparts in the IQC literature. This
connection is not pursued here but may lead to new insights within the
IQC framework.

\section{Notation}
\label{sec:notation}

$\R$ and $\C$ denote the set of real and complex numbers,
respectively.  $\RL$ denotes the set of rational functions with real
coefficients that are proper and have no poles on the imaginary axis.
$\RH$ is the subset of functions in $\RL$ that are analytic in the
closed right half of the complex plane.  $\C^{m\times n}$,
$\RL^{m\times n}$ and $\RH^{m\times n}$ denote the sets of $m\times n$
matrices whose elements are in $\R$, $\C$, $\RL$, $\RH$,
respectively. A single superscript index is used for vectors,
e.g. $\R^n$ denotes the set of $n\times 1$ vectors whose elements are
in $\R$.  For $z\in\C$, $\bar{z}$ denotes the complex conjugate of
$z$.  For a matrix $M\in \C^{m\times n}$, $M^T$ denotes the transpose
and $M^*$ denotes the complex conjugate transpose.  The para-Hermitian
conjugate of $G \in \RL^{m\times n}$, denoted as $G^{\sim}$, is
defined by $G^{\sim}(s):=G(-\bar{s})^*$.  Note that on the imaginary
axis, $G^{\sim}(j\omega)=G(j\omega)^*$.  $L_2^n[0,\infty)$ is the
space of functions $v: [0,\infty) \rightarrow \R^n$ satisfying
$\|v\|<\infty$ where
\begin{align}
  \|v\| := \left[ \int_0^\infty v(t)^Tv(t) \, dt \right]^{0.5}
\end{align}
Given $v\in L_2^n[0,\infty)$, $v_T$ denotes the truncated function:
\begin{align}
v_T(t) := \left\{ 
\begin{array}{ll} 
v(t) & \mbox{ for } t \le T \\
0 & \mbox{ for } t > T 
\end{array}
\right.
\end{align}
The extended space, denoted $L_{2e}$, is the set of functions $v$ such
that $v_T \in L_2$ for all $T \ge 0$. Finally, the Fourier Transform
$\hat{v} := \FT(v)$ maps the time domain signal $v\in
L_2^n[0,\infty)$ to the frequency domain by
\begin{align}
  \hat{v}(j\omega) := \int_0^\infty e^{-j\omega t} v(t) dt
\end{align}

\section{Problem Formulation}
\label{sec:probform}

Consider the time-delay system given by the feedback interconnection
of a nonlinear, time-invariant system $\tilde G$ and a (constant)
delay $\Dt$ as shown in \cref{fig:Dtfeedback}. The delay $\tilde
w=\Dt(v)$ is defined by $\tilde w(t) = v(t-\tau)$ where $\tau$
specifies the delay.  The input/output signals can, in general, be
vector-valued with $\tilde w(t),v(t) \in \R^{n_v}$. The remaining
signals in the interconnection have dimensions $d(t)\in\R^{n_d}$ and
$e(t)\in\R^{n_e}$. The feedback interconnection is obtained by closing
the upper channels of $\tilde G$ with the time delay $\Dt$. This
feedback interconnection, denoted as $F_u(\tilde G,\Dt)$, gives a
time-delay system with input $d$ and output $e$.

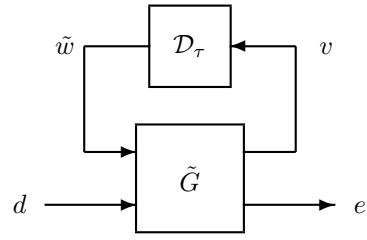
\begin{figure}[h!]
\centering
\scalebox{1.0}{
\begin{picture}(150,90)(20,20)
 \thicklines
 \put(75,25){\framebox(40,40){$\tilde{G}$}}
 \put(157,32){$e$}
 \put(115,35){\vector(1,0){35}}  
 \put(28,32){$d$}
 \put(40,35){\vector(1,0){35}}  
 \put(80,80){\framebox(30,30){$\Dt$}}
 \put(44,92){$\tilde w$}
 \put(55,55){\vector(1,0){20}}  
 \put(55,55){\line(0,1){40}}  
 \put(55,95){\line(1,0){25}}  
 \put(144,92){$v$}
 \put(135,95){\vector(-1,0){25}}  
 \put(135,55){\line(0,1){40}}  
 \put(115,55){\line(1,0){20}}  
\end{picture}
} 
\caption{Feedback interconnection with time delay $\Dt$}
\label{fig:Dtfeedback}
\end{figure}

A robust stability approach to analysis is pursued in this paper. Thus
it will be more convenient to express the system in terms of the
deviation between the delayed and the (nominal) undelayed signal,
$\St(v) := \Dt(v) - v$.  A loop transformation, shown in
\cref{fig:Stfeedback}, can be used to express the feedback
interconnection as $F_u(G,\St)$. This loop-shift amounts to the
replacement $\tilde{w} = w+v$ where $w:=\St(v)$. The system $G$
obtained after this loop-shift is assumed to be described by the
following finite-dimensional differential equation:
\begin{equation}
  \label{eq:G}
  \begin{split}
    \dot{x}_G & = f(x_G,w,d) \\
    v  & = h_1(x_G,w,d) \\
    e  & = h_2(x_G,w,d) 
  \end{split}
\end{equation}
where $x_G(t) \in \R^{n_G}$ is the state of $G$ at time $t$.

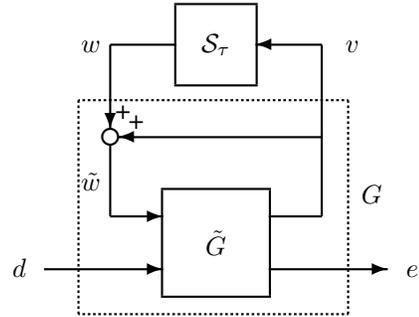
\begin{figure}[h!]
\centering
\scalebox{1.0}{
\begin{picture}(150,120)(20,15)
 \thicklines
 \put(75,25){\framebox(40,40){$\tilde{G}$}}
 \put(167,32){$e$}
 \put(115,35){\vector(1,0){45}}  
 \put(18,32){$d$}
 \put(30,35){\vector(1,0){45}}  
 \put(80,105){\framebox(30,30){$\St$}}
 \put(44,117){$w$}
 \put(55,120){\line(1,0){25}} 
 \put(144,117){$v$}
 \put(135,120){\vector(-1,0){25}}  
 \put(55,120){\vector(0,-1){32}}  
 \put(55,85){\circle{6}}
 \put(62,88){+}
 \put(57,92){+}
 \put(135,85){\vector(-1,0){77}}
 \put(55,82){\line(0,-1){27}}  
 \put(44,65){$\tilde w$}
 \put(55,55){\vector(1,0){20}}    
 \put(135,55){\line(0,1){65}}  
 \put(115,55){\line(1,0){20}} 
 \put(43,18){\dashbox(102,81){}}
 \put(150,60){$G$}
\end{picture}
} 
\caption{Loop transformation to $F_u(G,\St)$}
\label{fig:Stfeedback}
\end{figure}

An input-output approach is used in this paper to analyze the
stability and performance of the time-delay system. For a
given delay $\tau$, the induced $L_2$ gain for the feedback
interconnection from $d$ to $e$ is defined as:
\begin{align}
  \| F_u(G,\St) \| := \sup_{0\ne d \in L_2^{n_d}[0,\infty),  \ x_G(0)=0} 
      \frac{\|e\|}{\|d\|}
\end{align}
It is important to note that the restriction to time $t\ge 0$
implicitly assumes zero initial conditions for both $\Dt$ and $\St$.
Specifically, $\tilde w=\Dt(v)$ is more precisely defined on
$L_2[0,\infty)$ by $\tilde w(t)=0$ for $t \in [0,\tau)$ and $\tilde
w(t) = v(t-\tau)$ for $t\ge \tau$.  Similarly, $w=\St(v)$ is defined
on $L_2[0,\infty)$ by $w(t)=-v(t)$ for $t \in [0,\tau)$ and $ w(t) =
v(t-\tau)-v(t)$ for $t \ge \tau$. The notion of finite gain stability
used in this paper is defined next.

\vspace{0.1in}
\begin{defin}
  The feedback interconnection of $G$ and $\St$ is \underline{stable}
  if the interconnection is well-posed and if the mapping from $d$ to
  $e$ has finite $L_2$ gain, i.e. there exists a finite constant
  $\gamma>0$ such that $\|F_u(G,\St)\| \le \gamma$.
\end{defin}
\vspace{0.1in}


Two main analysis problems are considered.  First, determine the
largest value of $\bar{\tau}$ such that the feedback interconnection
is stable for all $\tau \in [0,\bar{\tau}]$. The first problem gives
the \emph{delay margin} for the system.  Second, given a delay $\tau$
less than the delay margin, determine the largest induced $L_2$ gain
from $d$ to $e$.  The second problem gives the performance of the
system for a fixed level of delay as measured by the $L_2$ gain.

Since the delay is constant, $\Dt$ defines a linear, time-invariant
(LTI) system. In this case the delay has a well-known frequency domain
representation. For constant delays $w=\Dt(v)$ can be expressed in the
frequency domain as $\hat{w}(j\omega) = \hatDt(j\omega)
\hat{v}(j\omega)$ where $\hatDt(j\omega) := e^{-j\omega \tau}$.  In
other words, the delay is equivalent to a frequency-by-frequency
multiplication by $\hatDt$.  Similarly, $\St$ has the frequency
response $\hatSt(j\omega) = e^{-j\omega \tau} - 1$.  These frequency
domain relations are used in the next section to derive simple,
geometric constraints satisfied by the input/output signals of $\St$.
Additional technical details on the frequency domain can be found in
standard textbooks, e.g. \cite{Dullerud99}.

Up to this point the presentation has focused entirely on constant
delays.  However, the analysis problems can be extended to consider
time-varying delays.  The time-varying delay $\tilde w=\Dtr(v)$ is
defined by $\tilde w(t) = v(t-\tau(t))$ where $\tau(t)$ specifies the
delay at time $t$.  The subscripts $\bar \tau$ and $r$ denote that the
delay satisfies $\tau(t) \in [0,\bar{\tau}]$ and $| \dot{\tau}(t) |
\le r$ for all $t \ge 0$. In other words, $\bar \tau$ is the maximum
delay and $r$ bounds the rate of variation.  If $r=0$ then $\Dtr$
corresponds to a constant delay with value $\tau \in [0,\bar \tau]$.
In addition, define $w=\Str(v)$ by $w = \Dtr(v)-v$, i.e. $\Str$ is the
deviation from the nominal undelayed signal.  A time-varying delay is
not a time-invariant system. Hence it does not have a valid
frequency-response interpretation. However, the frequency-domain
intuition can be used to derive constraints on the input/output
signals of a time-varying delay (\cref{sec:vtdqc}).

  

\section{Frequency Domain Inequalities}
\label{sec:freqdomain}

This section describes different frequency domain constraints on the
delay operator that can be incorporated into the input/output
analysis. The basic idea builds on common frequency domain
inequalities that have appeared in the literature,
e.g. \cite{Skogestad2005,Megretski1997,Gu2002}.  Given a constant delay
$\tau$, define the LTI system $\phi$ as:
\begin{equation}
\label{eq:w3}
\phi(s) := 2 \frac{(s {\tau})^2 + 3.5 s {\tau} + 10^{-6}}
            {(s {\tau})^2 + 4.5 s {\tau} + 7.1}.
\end{equation}
\cref{fig:staubounds1} shows the Bode magnitude plots for $\hatSt$
(solid line) and $\phi$ (dashed line) The weight $\phi$ is chosen to
satisfy $|\hatSt(j\omega)| \le |\phi(j\omega)|$ for all $\omega$ and
hence $\hatSt$ is a member of the following frequency weighted
uncertainty set:
\begin{equation}
  \label{eq:UncSet}
  \{ \Delta : |\Delta(j\omega)| \leq |\phi(j\omega)| \ \forall \omega \}
\end{equation} 
This magnitude bound has simple geometric and algebraic
interpretations at each frequency. The geometric interpretation of
this uncertainty set is a circle in the complex Nyquist plane as
depicted in \cref{fig:staubounds2}. $\hatSt(j\omega)$ follows a circle centered at $-1$ with radius $1$ (dashed circle). At each frequency
$\hatSt(j\omega)$ lies within the shaded circle of radius
$|\phi(j\omega)|$ centered at the origin.  An algebraic interpretation
can be given in terms of a quadratic constraint.  The Fourier
transforms for any input/output pair $w=\St(v)$ must satisfy the
following quadratic inequality at each frequency:
\begin{align}
\label{eq:QC1onSt}
  \bmtx \hat{v}(j \omega) \\ \hat{w}(j \omega) \emtx^* 
  \bmtx |\phi(j\omega)|^2 & 0 \\ 0 & -1 \emtx
  \bmtx \hat{v}(j \omega) \\ \hat{w}(j \omega) \emtx \geq 0
\end{align}
This quadratic constraint is just a restatement of the norm bound on
$\hatSt$ at each frequency. Using this basic intuition, additional
geometric constraints in the Nyquist plane can be expressed in the
form of quadratic constraints at each frequency.  Pointwise quadratic
constraints are discussed further in \cref{sec:qc} for general LTI
systems and the application to constant time delays is given in
\cref{sec:tdqc}.

\begin{figure}[h!]
\centering 
\subfigure[Bode Magnitude Plot for $\hatSt$ and bound $\phi$]
	{\input{figures/Staubound} \label{fig:staubounds1}}
\subfigure[Circle Interpretation for $|\hatSt(j\omega)|\le |\phi(j\omega)|$]
	{\includegraphics{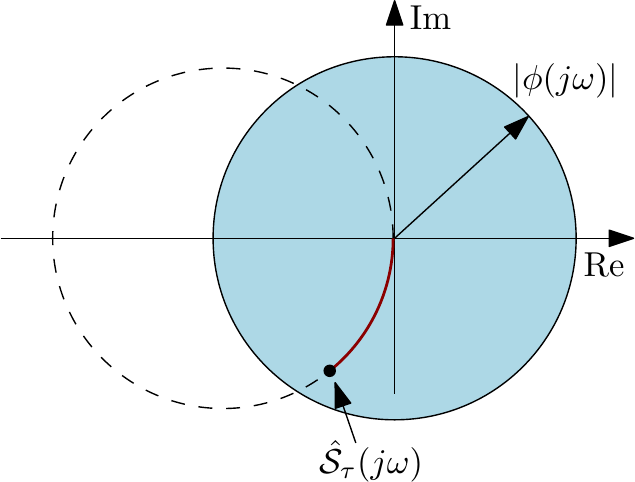} \label{fig:staubounds2}}
\caption{Norm Bound on $\hatSt$}
\label{fig:staubounds}
\end{figure}

\subsection{Pointwise Quadratic Constraints}
\label{sec:qc}

This section describes pointwise quadratic constraints for an LTI
system $\Delta$. For simplicity assume $\Delta$ is a single input,
single output (SISO) system.  The input/output relation $w=\Delta v$
is represented in the frequency domain by $\hat{w}(j\omega) =
\hat{\Delta}(j\omega)\hat{v}(j\omega)$. This representation can be
used to bound the input/output signals using frequency-by-frequency
quadratic constraints (QC).  

\vspace{0.1in}
\begin{defin}
\label{def:qc}
Let $\Pi: j \R \rightarrow \mathbb{C}^{2 \times 2}$ be a
Hermitian-valued function, called a "multiplier". Two signals $v,w\in
L_2[0,\infty)$ satisfy the QC defined by $\Pi$ if the following
inequality holds
\begin{align}
\label{eq:QC}
  \bmtx \hat{v}(j \omega) \\ \hat{w}(j \omega) \emtx^* 
  \Pi(j \omega) 
  \bmtx \hat{v}(j \omega) \\ \hat{w}(j \omega) \emtx \geq 0, \;
  \forall \omega
\end{align}
Moreover, the LTI system $\Delta$ satisfies the QC defined by $\Pi$,
denoted $\Delta \in \text{QC}(\Pi)$, if \cref{eq:QC} holds for all
$v\in L_2[0,\infty)$ and $w=\Delta v $.
\end{defin}
\vspace{0.1in}


These QCs have an intuitive geometric interpretation. Specifically,
the QCs can be interpreted as circle or half-plane constraints on the
Nyquist plot of $\hat{\Delta}(j\omega)$.  Partition the
$2\times 2$ multiplier as $\Pi= \bsmtx \pi_{11} & \pi_{21}^* \\
\pi_{21} & \pi_{22} \esmtx$. The diagonal entries are real since $\Pi$
is Hermitian. In addition, the QC is unaffected by positive scaling.
Specifically, let $\lambda:j\R \rightarrow \R$ be a frequency
dependent scaling that satisfies $\lambda(j\omega)>0$ $\forall
\omega$. Then $\Delta \in \text{QC}(\Pi)$ if and only if $\Delta \in
\text{QC}(\lambda\Pi)$.  Thus, without loss of generality, the
multiplier can be normalized to have $\pi_{22}(j\omega)=-1$, $0$, or
$+1$ for all $\omega$. This normalization provides a clearer geometric
interpretation for the QC as described in the next lemma.

\begin{lemma}
\label{lem:QC}
Let $\Delta$ be an LTI system satisfying $\Delta \in \text{QC}(\Pi)$.
At each frequency the QC can be normalized to one of three cases:
\begin{enumerate}
\item If $\pi_{22}(j\omega)=-1$ then
\begin{align}
  \label{eq:QCcircle}
  | \hat{\Delta}(j\omega) - \pi_{21}(j\omega) | 
     \le \sqrt{ \pi_{11}(j\omega) + |\pi_{21}(j\omega)|^2}
\end{align}
\item If $\pi_{22}(j\omega)=0$ then
\begin{align}
  \label{eq:QChp}
  0 \le \pi_{11}(j\omega) + \hat{\Delta}(j\omega)^* \pi_{21}(j\omega) + 
              \pi_{21}(j\omega) \hat{\Delta}(j\omega)     
\end{align}
\item If $\pi_{22}(j\omega)=+1$ then
\begin{align}
  \label{eq:QCcircle2}
  | \hat{\Delta}(j\omega) - \pi_{21}(j\omega) | 
     \ge \sqrt{ \pi_{11}(j\omega) + |\pi_{21}(j\omega)|^2}
\end{align}
\end{enumerate}
\end{lemma}
\begin{proof}
  The QC in Equation~\ref{eq:QC} must hold for all input signals $v
  \in L_2[0,\infty)$.  Since $\hat{w}(j\omega) =
  \hat{\Delta}(j\omega)\hat{v}(j\omega)$, the QC can be rewritten as
  (dropping the dependence on $j\omega$):
  \begin{align}
    \label{eq:QC2}
    0 \le \bmtx 1 \\ \hat{\Delta} \emtx^* \Pi \bmtx 1 \\ \hat{\Delta}
    \emtx = \pi_{11} + \hat{\Delta}^* \pi_{21} + \pi_{21}^* \hat{\Delta} -
    \hat{\Delta}^* \pi_{22} \hat{\Delta}
  \end{align}
  If $\pi_{22}(j\omega)=0$ then this simplifies to
  Equation~\ref{eq:QChp}. If $\pi_{22}(j\omega)=\pm 1$ then complete
  the square to express Equation~\ref{eq:QC2} as in
  Equation~\ref{eq:QCcircle} or Equation~\ref{eq:QCcircle2}.
\end{proof}

Equation~\ref{eq:QCcircle} defines a circle and its interior (a disk)
in the complex plane centered at $\pi_{21}(j\omega)$ with radius given
by $\sqrt{ \pi_{11}(j\omega) + |\pi_{21}(j\omega)|^2}$.
Equation~\ref{eq:QChp} defines a half plane in the complex plane.  For
example if $\pi_{21}(j\omega)=1$ then \cref{eq:QChp} defines the half
plane given by $\text{Re}(\hat{\Delta}(j\omega)) \ge -\frac{1}{2}
\pi_{11}(j\omega)$.  Finally, the case where $\pi_{22}(j\omega)=+1$
corresponds to the non-convex set described by a circle and its
exterior.  Thus a QC defines a circle or half-plane constraint in the
complex (Nyquist) plane at each frequency.  The dissipation theory
developed below is only valid for $\pi_{22}(j\omega)<0$ and hence the
focus will be on quadratic constraints that are described by disks.
``G''-scales used for robustness analysis with real parameter
uncertainty can also be interpreted as a circle constraint
\cite{doyle85,fan91}.

Multiple QCs can be combined to obtain new QCs. If the LTI system
$\Delta$ satisfies the QCs defined by $\{ \Pi_k \}_{k=1}^N$ then
$\Delta$ also satisfies the QC defined by $\Pi(\lambda):= \sum_{k=1}^N
\lambda_k \Pi_k$ for any real, non-negative numbers
$\{\lambda_k\}_{k=1}^N$.  $\Pi(\lambda)$ is called a conic combination
of the multipliers $\{ \Pi_k \}_{k=1}^N$.  This fact enables many QCs
on $\Delta$ to be incorporated into an analysis.  However, it is
important to recognize that such conic combinations have certain
limitations.  The main limitation is that any conic combination is, by
\cref{lem:QC}, just a circle or half-plane constraint in the complex
plane.  As a concrete example, consider the circle constraint on $\St$
described by Equation~\ref{eq:QC1onSt} and shown in
\cref{fig:staubounds2}.  The corresponding multiplier is
$\Pi_1(j\omega) := \bsmtx |\phi(j\omega)|^2 & 0 \\ 0 & -1 \esmtx$.
$\St$ also satisfies the QC defined by $\Pi_2 = \bsmtx 0 & -1 \\ -1 &
-1 \esmtx$. This second multiplier corresponds to a circle centered at
$-1$ with unit radius, shown as the dashed circle in
\cref{fig:staubounds2}.  Thus $\St$ lies in the intersection of the
shaded circle centered at the origin (defined by $\Pi_1$) and the
dashed circle centered at $-1$ (defined by $\Pi_2$). However, the
conic combinations $\Pi(\lambda) := \lambda_1 \Pi_1 + \lambda_2 \Pi_2$
correspond to circles that ``cover'', i.e. outer bound, this
intersection.

\vspace{0.1in} 
\begin{remark}
  In fact, conic combinations of two QC multipliers are ``tight''.  Roughly,
  the conic combinations $\lambda_1 \Pi_1 + \lambda_2 \Pi_2$ define
  the smallest regions that contain the intersection of the sets
  defined by $\Pi_1$ and $\Pi_2$.  This statement is a geometric
  consequence of the S-procedure lossless theorem for complex
  constraints \cite{fradkov79,Dullerud99,jonsson06}.  However the
  S-procedure is not lossless, in general, for three or more QCs.
  Hence conic combinations need not provide a tight bound on the set
  described by the intersection of three or more multipliers.
\end{remark}
\vspace{0.1in} 

The QCs defined above for SISO systems can be extended, with only
notational changes, to multiple-input, multiple output (MIMO) systems.
It will be sufficiently general for the time-delay analysis to
consider repeated systems. Let $\Delta$ be a SISO, LTI system and
define the $n \times n$ repeated system $w=(\Delta \cdot I_n)(v)$ by
$w_i = \Delta v_i$ for $i=1,\ldots,n$.  If $\Delta \in \text{QC}(\Pi)$
for a $2 \times 2$ multiplier $\Pi$ then the following QC holds for
all $v\in L_2^n[0,\infty)$ and $w=(\Delta \cdot I_n) v$:
\begin{align}
\label{eq:QCmimo}
  \bmtx \hat{v}(j \omega) \\ \hat{w}(j \omega) \emtx^* 
  \bmtx \pi_{11}(j\omega) \cdot I_n & \pi_{12}(j\omega) \cdot  I_n \\
        \pi_{21}(j\omega) \cdot I_n & \pi_{22}(j\omega) \cdot I_n \emtx
  \bmtx \hat{v}(j \omega) \\ \hat{w}(j \omega) \emtx \geq 0, \;
  \forall \omega
\end{align}

Moreover $\Delta$ is LTI and hence $\Delta \cdot I_n$ commutes with
any $n\times n$, frequency-dependent matrix $D$, i.e. $D (\Delta \cdot
I_n) = (\Delta \cdot I_n) D$.  Thus the frequency-scaled system
$\bar{\Delta}:=D \Delta D^{-1}$ also satisfies the QC in
Equation~\ref{eq:QCmimo}.  Let $(\bar{v},\bar{w})$ be any input-output
pair for the scaled system $\bar{\Delta}$ as shown in
\cref{fig:dscales}. The associated input/output pair for the original
system $w= (\Delta \cdot I_n) v$ is related to the input/output pair
for the scaled system by $\bar{w}=Dw$ and $\bar{v}=Dv$. Hence, $\Delta
\cdot I_n$ also satisfies the QC with any multiplier 
$\Pi_n: j \R \rightarrow \mathbb{C}^{2n \times 2n}$ of the form
\begin{align}
\label{eq:QCmimo2}
\Pi_n(j\omega) :=
  \bmtx 
  \pi_{11}(j\omega) \cdot X(j\omega) & \pi_{12}(j\omega) \cdot X(j\omega) \\
  \pi_{21}(j\omega) \cdot X(j\omega) & \pi_{22}(j\omega) \cdot X(j\omega) 
  \emtx
\end{align}
where $X(j\omega):= D(j\omega)^* D(j\omega) \ge 0$.  This more
general, frequency-scaled multiplier can be used to reduce the
conservatism in the analysis. The use of $X$ is analogous to
the multipliers used in classical robustness analysis, e.g. the
structured singular value $\mu$
\cite{Safonov1980,Doyle1982,Packard1993,zhou96}.

\begin{figure}[h!]
  \centering
  \includegraphics{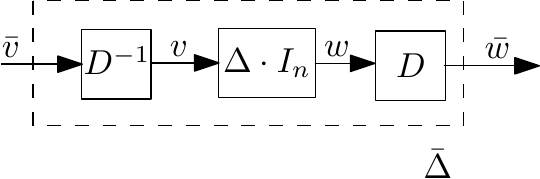}
  \caption{Scaling of the Operator $\Delta \cdot I_n$}
  \label{fig:dscales}
\end{figure}

\subsection{Application to Constant Time Delays}
\label{sec:tdqc}

There are numerous circle constraints that can be used to bound
$\St$. A few examples are provided below to illustrate how the
frequency-domain geometric interpretation can be used to impose
constraints on $\St$. The QCs on $\St$ can be converted, if needed,
into equivalent QCs on $\Dt$ by reversing the loop-transformation,
i.e. by replacing $w=\tilde{w}-v$ in the QC. For clarity, the QC
multipliers are given assuming $\St$ is SISO. However, as described
previously, frequency-dependent (matrix) scalings can be introduced if
$\St$ is MIMO.

\vspace{0.1in} 
\begin{ex}
  The Nyquist plot for $\hatSt$ follows a circle centered at $-1$ with
  radius $1$. Hence, the most basic QC to bound $\St$ describes
  exactly this circle. This corresponds to the following multiplier:
  \begin{align}
    \label{eq:Pi1}
    \Pi_1 := \bmtx 0 & -1 \\ -1 & -1 \emtx.
  \end{align}
\end{ex} 

\vspace{0.1in}
\begin{ex}
  The multiplier in Example 1 does not depend on the value of the time
  delay $\tau$. Thus the multiplier in \cref{eq:Pi1} describes a very
  conservative constraint due to this delay independence.  A simple
  delay dependent constraint is obtained from the frequency response
  $\hatSt(j\omega)=e^{-j\omega\tau} -1$.  Thus at each frequency
  $\hatSt(j\omega)$ lies within a circle centered at the origin of
  radius $|\hatSt(j\omega)|$.  The multiplier $\Pi_2$ for this circle
  is given by:
  \begin{align}
    \label{eq:Pi2}
    \Pi_2(j\omega) := \bmtx \abs{\hatSt(j\omega)}^2 & 0 \\ 0 & -1 \emtx
  \end{align}
\end{ex}

\vspace{0.1in}
\begin{ex}
  A smaller circle constraint can be constructed for $\St$.  The
  midpoint of the segment connecting $\hatSt(j\omega)$ and the origin
  is given by $\frac{1}{2}\hatSt(j\omega)$. The following multiplier
  $\Pi_3$ defines a circle centered at this midpoint with radius equal
  to the absolute value of this midpoint.
  \begin{align}
    \label{eq:Pi3}
    \Pi_3(j\omega) & := 
             \bmtx 0 & \frac{1}{2}\hatSt(j\omega) \\
               \frac{1}{2}\hatSt(j\omega) & -1 \emtx 
  \end{align}
  The QC described by $\Pi_3$ is shown in \cref{fig:circleathp}.
  \begin{figure}[h!]
    \centering
    \includegraphics{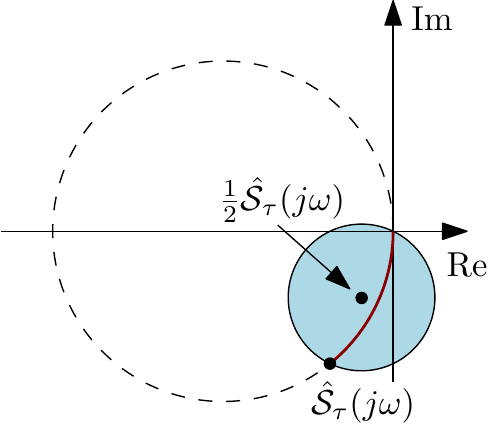}
    \caption{``Small'' Circle Constraint on $\St$ described by QC($\Pi_3$)}
    \label{fig:circleathp}
  \end{figure}
\end{ex}
\vspace{0.1in}


As can be seen by the examples, there exist numerous ways to bound
$\St$ using a QC. More examples are given in literature. For instance,
in \cite{Megretski1997} a multiplier is provided that corresponds to a
circle whose center moves along the imaginary axis with frequency.
Unfortunately there are no general rules for which specific QC will
provide the most useful stability and performance analysis results.
It might seem intuitive to use the QC that describes the smallest area
circle.  However, this does not, in general, provide the least
conservative analysis results.  The best solution is to specify many
different QCs and allow a numerical algorithm select the best conic
combination.  This will be described further in the next section.

Two practical issues must be addressed to make these QCs useful for
numerical analysis.  First, recall that a system $F_u(G,\St)$ has
delay margin of $\bar \tau$ if it is stable for all delays $\tau \in
[0,\bar \tau]$.  Thus the delay margin analysis requires a QC that
covers $\St$ for all $\tau \in [0,\bar \tau]$.  For example, the
following multiplier covers $\St$ for all constant delays $\tau \in
[0,\bar \tau]$:
\begin{align}
  \label{eq:Pi2til}
  \tilde{\Pi}_2(j\omega) := \begin{cases}
      \bmtx \abs{\hatStbar(j\omega)}^2 & 0 \\ 0 & -1 \emtx
      & \text{if } \omega \le \frac{\pi}{\bar \tau} \\ \\
      \bmtx 4 & 0 \\ 0 & -1 \emtx 
      & \text{else}
    \end{cases}
\end{align}
At each low frequency ($\omega \le \frac{\pi}{\bar \tau}$),
$\tilde{\Pi}_2$ describes a circle centered at the origin that covers
all Nyquist curves $\hatSt(j\omega)$ for $\tau \in [0,\bar \tau]$.
This is essentially equivalent to the multiplier in Example 2 at low
frequencies.  However, these circles fail to have the desired covering
property at higher frequencies.  For $\omega > \frac{\pi}{\bar \tau}$,
the multiplier $\tilde{\Pi}_2$ is set equal to a circle centered at
the origin of radius 2.  This ensures that the QC defined by
$\tilde{\Pi}_2$ covers the entire Nyquist curve of $\hatSt$ at high
frequencies. The other QC multipliers specified in the examples can be
similarly modified at high frequencies for use in delay margin
analysis.

The second practical issue is the need to approximate a QC multiplier
with a rational function so that state-space numerical methods can be
applied. For example $\tilde{\Pi}_2$ is a non-rational multiplier that
describes a circle at each frequency.  Let $\phi_2(s)$ be any stable
LTI system that satisfies $|\phi_2(j\omega)| \ge |\hatStbar(j\omega)|$
for $\omega \le \frac{\pi}{\bar \tau}$ and $|\phi_2(j\omega)|\ge 2$
for $\omega > \frac{\pi}{\bar \tau}$. The specific choice in
Equation~\ref{eq:w3} satisfies these constraints. Then $\Stbar$
satisfies the QC defined by the following (rational) multiplier for
all $\tau \in [0,\bar \tau]$:
\begin{align}
  \label{eq:Pi2bar}
  \bar{\Pi}_2(j\omega) & = \bmtx |\phi_2(j\omega)|^2 & 0 \\ 0 & -1 \emtx
\end{align}
$\bar{\Pi}_2(j\omega)$ describes a circle with larger radius than
$\tilde{\Pi}_2(j\omega)$ due to the choice of $\phi_2(s)$.  Hence $\St
\in \text{QC}(\tilde{\Pi}_2) \subset \text{QC}(\bar{\Pi}_2)$ for all
$\tau \in [0,\bar \tau]$.  $\bar{\Pi}_2$ is the weighted multiplier
(Equation~\ref{eq:QC1onSt}) discussed earlier in the section. It is
commonly used in the literature to analyze time delayed systems
\cite{Skogestad2005,Megretski1997,Gu2002}.  

\vspace{0.1in}
\begin{remark}
  Numerical tools can be used to aid the construction of rational
  function approximations for more complicated multipliers.  Briefly
  consider a (non-rational) multiplier $\Pi = \bsmtx \pi_{11} & \pi_{21}^* \\
  \pi_{21} & -1 \esmtx$.  By \cref{lem:QC} the QC associated with this
  multiplier defines a circle centered at $\pi_{21}(j\omega)$ with
  radius given by $\sqrt{ \pi_{11}(j\omega) + |\pi_{21}(j\omega)|^2}$.
  A rational approximation $\bar \Pi$ can be computed by the following
  procedure. Fit $\pi_{21}$ with a rational approximation
  $\bar{\pi}_{21}$, e.g. using the Matlab function \texttt{fitfrd} or
  some other numerical optimization. This approximation will introduce
  some fitting error and hence the radius of the rational multiplier
  may need to be increased to ensure $\St$ is covered.  This would
  require $\bar{\pi}_{11}$ to satisfy a lower bound constraint on its
  magnitude.  A rational approximation with a lower bound magnitude
  constraint can be computed for $\bar{\pi}_{11}$, e.g. using the
  Matlab function \texttt{fitmagfrd}. Using this procedure, the
  following rational fit was constructed for $\Pi_3$ that covers $\St$
  for all delays $\tau \in [0,\bar \tau]$
  \begin{align}
    \label{eq:Pi3bar}
    \bar{\Pi}_{3}(j\omega):= \bmtx 0 & \phi_3^*(j\omega)
            \\ \phi_3(j\omega)& -1 \emtx
  \end{align}
  where
  \begin{align}
    \phi_3(j\omega) := 
      \frac{-2.19 \left( \frac{j\omega}{\bar{\tau}} \right)^2 
           + 9.02 \left( \frac{j\omega}{\bar{\tau}} \right) + 0.089} 
      {\left( \frac{j\omega}{\bar{\tau}} \right)^2  
           - 5.64 \left( \frac{j\omega}{\bar{\tau}} \right) - 17.0}.
  \end{align}
  
 Note that $\pi_{21}$ does not need to be stable, as is the case with the given choice of $\phi_3$.
\end{remark}
\vspace{0.1in}



\subsection{Application to Varying Time Delays}
\label{sec:vtdqc}

The QCs given for constant time delays hold at each frequency.  These
point-wise QCs fall within a more general framework based on integral
quadratic constraints (IQCs) introduced in \cite{Megretski1997}. A
definition is now given for IQCs that extends the one given in
\cref{def:qc} for QCs.

\vspace{0.1in}
\begin{defin}
\label{def:iqc}
Let $\Pi: j\R \rightarrow \C^{(m_1+m_2)\times(m_1+m_2)}$ be a
Hermitian-valued function.  Two signals $v \in L_2^{m_1}[0,\infty)$
and $w \in L_2^{m_2}[0,\infty)$ satisfy the integral quadratic
constraint (IQC) defined by $\Pi$ if
\begin{align}
\label{eq:iqccond}
\int_{-\infty}^{\infty} 
\bsmtx \hat{v}(j\omega) \\ \hat{w}(j\omega) \esmtx^*
 \Pi(j\omega)
\bsmtx \hat{v}(j\omega) \\ \hat{w}(j\omega) \esmtx
d\omega
\ge 0
\end{align}
where $\hat{v}(j\omega)$ and $\hat{w}(j\omega)$ are Fourier transforms
of $v$ and $w$, respectively. A bounded, causal operator $\Delta :
L_{2e}^{m_1}[0,\infty) \rightarrow L_{2e}^{m_2}[0,\infty)$ satisfies
the IQC defined by $\Pi$, denoted $\Delta \in \text{IQC}(\Pi)$, if
\cref{eq:iqccond} holds for all $v \in L_2^{m_1}[0,\infty)$ and
$w=\Delta(v)$.
\end{defin}
\vspace{0.1in}

Clearly, if the QC holds pointwise in frequency it also holds when
integrated over all frequencies, i.e. $\Delta \in \text{QC}(\Pi)$
implies $\Delta \in \text{IQC}(\Pi)$. The converse is not true and the
more general IQC theory can be applied to nonlinear and/or time
varying perturbations.  IQCs were introduced in \cite{Megretski1997}
to provide a general framework for robustness analysis.  Here, the
focus will be on the use of IQCs to describe the input/output behavior
of time-varying delays.  A time-varying delay is not a time-invariant
system. Hence the point-wise QCs cannot be used and the more general
IQCs are required.  However, the frequency domain arguments used to
construct QCs for constant delays can provide intuition regarding IQCs
for time-varying delays.  The remainder of this section reviews known
IQCs for time-varying delays \cite{Kao2004,Kao2007,Gu2002}. In
addition a new IQC for time-varying delays is constructed using the
frequency domain intuition gained from constant-delays.

Recall the notation $\Dtr$ and $\Str$ introduced for varying delays
where $\bar \tau$ and $r$ are bounds on the maximum delay and rate of
variation, respectively.  The basic IQCs for time-varying delays arise
from two simple norm bounds.  First, if $r<1$ then the induced $L_2$
gain of the varying-time delay can be bounded by $\| \Dtr \| \le
\frac{1}{\sqrt{1-r}}$ (Section 3.2 in \cite{Gu2002} and Lemma 1 in
\cite{Kao2007}). Second, let $\Str \circ \frac{1}{s}$ denote $\Str$
composed with an integrator at the input.  The induced $L_2$ gain of
this combined system can be bounded by $\| \Str \circ \frac{1}{s} \|
\le \bar \tau$ (Lemma 1 in \cite{Kao2004}). These two bounds are tight
in the sense that the gain is achieved for some input signal $v$ and
time-varying delay $\tau(t)$ that satisfies the bounds $\bar \tau$ and
$r$ (Lemma 1 in \cite{Kao2007}).  Three IQCs are now provided for
time-varying delays.  These examples parallel the QCs provided in the
previous subsection for constant delays. For clarity the multipliers
are expressed assuming $\Dtr$ and $\Str$ are SISO. The extension to the
MIMO case is discussed below.

\vspace{0.1in} 
\begin{ex}
  Let $\tilde w=\Dtr(v)$ be a time-varying delay satisfying $r<1$. The
  bound $\| \Dtr \| \le \frac{1}{\sqrt{1-r}}$ implies that $\| \tilde
  w\| \le \frac{1}{\sqrt{1-r}} \| v\|$ for all input signals $v\in
  L_2[0,\infty)$. After performing the loop-transformation
  $w:=\Str(v)$, i.e. $\tilde w = w+v$, this inequality can be written
  as:
  \begin{align}
    \int_0^\infty \bmtx v(t) \\ w(t) \emtx^T
       \bmtx \frac{r}{1-r}  & -1 \\ -1 & -1 \emtx
       \bmtx v(t) \\ w(t) \emtx dt \ge 0
  \end{align}
  By Parseval's theorem, this inequality can be equivalently expressed 
  in the frequency domain as:
  \begin{align}
    \int_{-\infty}^{\infty} 
    \bmtx \hat{v}(j\omega) \\ \hat{w}(j\omega) \emtx^*
    \bmtx \frac{r}{1-r}  & -1 \\ -1 & -1 \emtx
    \bmtx \hat{v}(j\omega) \\ \hat{w}(j\omega) \emtx
    d\omega
    \ge 0
  \end{align}
  Thus $\Str$ satisfies the IQC defined by the following multiplier:
  \begin{align}
    \label{eq:Pi4}
    \Pi_4 :=     \bmtx \frac{r}{1-r}  & -1 \\ -1 & -1 \emtx
  \end{align}
  This multiplier is analogous to the multiplier $\Pi_1$ given in
  Example 1 for constant delays.  $\Pi_1$ corresponds to a unit circle
  in the Nyquist plane centered at $-1$. At each frequency
  $\Pi_4$ can also be viewed as a circle centered at $-1$ but
  with radius enlarged to $\frac{1}{\sqrt{1-r}}>1$ to account for the
  varying delays.  However this is not a precise statement since the
  constraints defined by $\text{IQC}(\Pi_4)$ only hold when integrated
  over all frequencies.  
\end{ex} 

\vspace{0.1in}
\begin{ex}
  The multiplier in Example 4 depends on the rate of variation $r$
  but does not depend on the maximum delay $\bar \tau$. Proposition 2
  in \cite{Kao2007} gives a delay-dependent IQC that can be used to
  reduce the conservatism.  The IQC in \cite{Kao2007} depends on a
  rational bounded transfer function $\phi_5(s)$ that satisfies:
  \begin{align}
    |\phi_5(j\omega)| > \left\{ 
         \begin{array}{ll}
               \bar \tau | \omega | & \text{if } 
               \bar \tau |\omega| \le 1 + \frac{1}{\sqrt{1-r}} \\
               1+ \frac{1}{\sqrt{1-r}} & \text{if } 
               \bar \tau |\omega| > 1 + \frac{1}{\sqrt{1-r}} 
               \end{array}
             \right.
  \end{align}
  If $r<1$ then $\Str$ satisfies the IQC defined by the following
  multiplier $\Pi_5$:
  \begin{align}
    \label{eq:Pi5}
    \Pi_5(j\omega) := \bmtx |\phi_5(j\omega)|^2  & 0 \\ 0 & -1 \emtx
  \end{align}
  The proof that $\Str \in \text{IQC}(\Pi_5)$ uses the bound $\| \Str
  \circ \frac{1}{s} \| \le \bar \tau$.  The IQC multiplier $\Pi_5$ is
  analogous to the QC multiplier $\Pi_2$ provided in Example 2 for
  constant delays. Specifically, a Taylor series expansion for
  $\hatStbar$ can be used to show that $\Pi_5$ is equivalent to
  $\Pi_2$ at low frequencies.  The bound on $|\phi_5|$ effectively
  increases the radius of circle constraints defined by $\Pi_5$ at
  high frequencies to account for the time variations and to cover all
  delays in $[0,\bar \tau]$.  Note that as $r\rightarrow 0$, the high
  frequency bound in $\Pi_5$ becomes $|\phi_5(j\omega)| > 2$.  Thus as
  $r\rightarrow 0$, $\Pi_5$ converges to the multipliers
  $\tilde{\Pi}_2$ used to cover all constant delays in $[0,\bar
  \tau]$.  Again, these interpretations of $\Pi_5$ are imprecise and
  only meant to provide an intuitive interpretation. Proposition 3 in
  \cite{Kao2007} gives a similar IQC multiplier that is valid for
  $r<2$.
\end{ex}

\vspace{0.1in}
\begin{ex}
  Example 3 in \cref{sec:tdqc} provided a QC multiplier $\Pi_3$ for
  constant delays corresponding to a circle centered at $\frac{1}{2}
  \hatSt(j\omega)$.  The benefit of this multiplier is that it defined
  a smaller circle than the multipliers given in Examples 1 and 2.
  This frequency domain intuition can be used to derive a new, related
  IQC for time-varying delays.  The new IQC depends on a rational
  transfer function $\phi_6(s)$ that satisfies:
  \begin{align}
    \label{eq:phi6}
    |\phi_6(j\omega)| > \left\{ 
         \begin{array}{ll}
               \frac{1}{2} \bar \tau | \omega | & \text{if } 
               \frac{1}{2} \bar \tau |\omega| \le 1 + \frac{1}{\sqrt{1-r}} \\
               1+ \frac{1}{\sqrt{1-r}} & \text{if } 
               \frac{1}{2} \bar \tau |\omega| > 1 + \frac{1}{\sqrt{1-r}} 
               \end{array}
             \right.
  \end{align}
  It is shown in Appendix~\ref{app:newiqc} that if $r<1$ then $\Str$
  satisfies the IQC defined by the following multiplier
  $\Pi_6$:
 \begin{align}
   \label{eq:Pi6}
   \Pi_6(j\omega)  := 
   \bmtx |\phi_6(j\omega)|^2 - \frac{1}{4} | \hatStbar(j\omega)|^2
   & \frac{1}{2} \hatStbar(j\omega) \\
   \frac{1}{2} \hatStbar(j\omega) & -1 \emtx 
 \end{align}
 The proof that $\Str \in \text{IQC}(\Pi_6)$ uses the bound $\| (\Str
 -\frac{1}{2} \Stbar) \circ \frac{1}{s} \| \le \frac{1}{2} \bar \tau$.
 Again, a Taylor series expansion for $\hatStbar$ can be used to show
 that $\Pi_6$ is equivalent to the analogous multiplier for constant
 delays $\Pi_3$ at low frequencies.  The bound on $|\phi_6|$
 effectively increases the radius of circle constraints defined by
 $\Pi_6$ at high frequencies to account for the time variations and to
 cover all delays in $[0,\bar \tau]$.
\end{ex}
\vspace{0.1in}

Each of the multipliers can be generalized for the case where $\Str$
is MIMO using the idea of scalings already introduced in
Section~\ref{sec:qc}.  $\Str$ is not a time-invariant system and hence
frequency-dependent scalings cannot be used. However, the linearity of
$\Str$ can be used to show that constant matrix scalings can be
introduced into the multiplier.  For example, $\Pi_4$ remains
a valid IQC multiplier for $\Str$ if it is generalized to include
any matrix $X \ge 0$:
\begin{align}
  \label{eq:Pi4X}
  \Pi_4 := \bmtx \frac{r}{1-r} X & -X \\ -X & -X \emtx.
\end{align}
A formal proof that $\Str \in \text{IQC}(\Pi_4)$ is given in
Proposition 1 of \cite{Kao2007}.  Reference \cite{Kao2007} derives
additional IQCs for time-varying delays.  In particular,
frequency-dependent scalings can be introduced into the multipliers
for varying delays but, in this case, a swapping lemma must be used to
account for the the time-variations in the delay. This section does
not intend to provide an exhaustive review of IQCs for time-varying
delays.  Instead the main purpose is to demonstrate the benefit of the
frequency-domain intuition provided by QCs for constant delays.
Example 6 is a new IQC derived for time-varying delays using this
intuition and it should be possible to derive additional useful IQCs
using this approach.

\section{Time Domain Stability Analysis}
\label{sec:timedomain}
%

The previous section defined frequency domain constraints that
describe the input/output behavior of a time delay.  These were
specified as QCs for constant delays and IQCs for time-varying delays.
This section shows that, under some mild technical conditions, these
constraints have an equivalent time domain representation
(\cref{sec:tdIQC}).  The alternative time domain representation for
QCs and IQCs is used to derive dissipation inequality based
stability conditions for delayed nonlinear and parameter varying
systems (Sections~\ref{sec:nlstab} and~\ref{sec:lpvstab}).

\subsection{Time Domain IQCs}
\label{sec:tdIQC}


Let $\Pi$ be an QC or IQC multiplier that is a rational and uniformly
bounded function of $j\omega$, i.e. $\Pi \in \RL^{(m_1+m_2) \times
  (m_1+m_2)}$. As noted previously QCs hold pointwise in frequency and
hence they are also valid when integrated over frequency. In other
words, if $\Pi$ is a valid QC multiplier for a  system $\Delta$
then it is also a valid IQC multiplier for the same system.  Thus it
is sufficient to provide a time domain interpretation for IQCs.

The time domain interpretation is based on factorizing the multiplier
as $\Pi = \Psi^{\sim} M \Psi$ where $M=M^T \in \R^{n_z \times n_z}$
and $\Psi \in \RH^{n_z \times(m_1+m_2)}$. The restriction to rational,
bounded multipliers $\Pi$ ensures that such factorizations can be
numerically computed via transfer function or state-space methods
\cite{youla61,Scherer2004,materassi09}. Such factorizations are not
unique and two specific factorizations are provided in
Appendix~\ref{app:IQCfac} using state-space methods.  For a general
factorization, $\Psi$ is assumed to be stable but may be non-square
($n_z\ne m$) and possibly non-minimum phase.

Next, let $(v,w)$ be a pair of signals that satisfy the IQC in
\cref{eq:iqccond} and define
$\hat{z}(j\omega):= \Psi(j \omega) \bsmtx \hat{v}(j\omega) \\
\hat{w}(j\omega) \esmtx$.  Then the IQC can be written as:
\begin{align}
\label{eq:iqccond2}
\int_{-\infty}^{\infty} 
\hat{z}(j\omega)^*
M
\hat{z}(j\omega)
d\omega
\ge 0
\end{align}
By Parseval's theorem \cite{zhou96}, this frequency-domain inequality
can be equivalently expressed in the time-domain as:
\begin{align}
\label{eq:iqctd}
\int_0^\infty z(t)^T M z(t) \, dt \ge 0
\end{align}
where $z$ is the output of the LTI system $\Psi$:

\begin{equation}
  \label{eq:ltipsi}
  \begin{split}
    \dot{\psi}(t) & = A_\psi \psi(t) + B_{\psi 1} v(t) + B_{\psi 2} w(t) 
    , \;\;\; \psi(0) = 0 \\
    z(t) & = C_\psi \psi(t) + D_{\psi 1} v(t) + D_{\psi 2} w(t) 
  \end{split}
\end{equation}
Thus signals $v \in L_2^{m_1}[0,\infty)$ and $w \in
L_2^{m_2}[0,\infty)$ satisfy the IQC defined by $\Pi=\Psi^\sim M \Psi$
if and only if the filtered signal $z=\Psi \bsmtx v \\ w \esmtx$
satisfies the time domain constraint in \cref{eq:iqctd}.
Similarly, a bounded, causal system $\Delta$ satisfies the IQC defined
by $\Pi=\Psi^\sim M \Psi$ if and only if \cref{eq:iqctd} holds
for all $v \in L_2^{m_1}[0,\infty)$ and $w=\Delta(v)$. To simplify
notation, $\Delta \in \text{IQC}(\Pi)$ will also be denoted by $\Delta
\in \text{IQC}(\Psi,M)$. \cref{fig:iqcpsi} provides a graphical
interpretation for this time-domain form of $\Delta \in
\text{IQC}(\Psi,M)$. The input and output signals of $\Delta$ are
filtered through $\Psi$ and the output $z$ satisfies the time-domain
inequality in \cref{eq:iqctd}.  A simple example is provided
to illustrate the connection between the time domain and frequency
domain constraints.

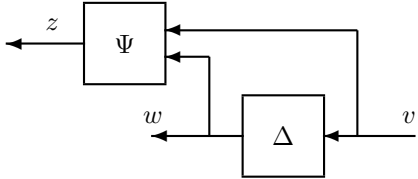
\begin{figure}[h]
\centering
\scalebox{1.0}{
\begin{picture}(150,65)(-10,80)
 \thicklines
 \put(80,80){\framebox(30,30){$\Delta$}}
 \put(42,100){$w$}
 \put(80,95){\vector(-1,0){35}}  
 \put(140,100){$v$}
 \put(145,95){\vector(-1,0){35}}  
 \put(67,95){\line(0,1){30}}
 \put(67,125){\vector(-1,0){17}}
 \put(123,95){\line(0,1){40}}
 \put(123,135){\vector(-1,0){73}}
 \put(20,115){\framebox(30,30){$\Psi$}}
 \put(5,135){$z$}
 \put(20,130){\vector(-1,0){30}}
\end{picture}
} 
\caption{Graphical interpretation of the IQC defined by $\Pi=\Psi^\sim M \Psi$}
\label{fig:iqcpsi}
\end{figure}

\begin{ex}
  Consider the multiplier $\bar{\Pi}_2(j\omega) = \bsmtx
  |\phi_2(j\omega)|^2 & 0 \\ 0 & -1 \esmtx$
  (Equation~\ref{eq:Pi2bar}) where $\phi_2$ is a stable LTI
  system that satisfies certain bounds.  This multiplier is a rational
  approximation that corresponds, at low frequencies, to a circle
  centered at the origin with a frequency-dependent radius.  $\St \in
  \text{QC}(\bar{\Pi}_2)$ implies $\St \in
  \text{IQC}(\bar{\Pi}_2)$. Hence any $v=L_2[0,\infty)$ and $w = \St
  v$ satisfy
  \begin{align}
    \int_{-\infty}^{\infty} 
    \bsmtx \hat{v}(j\omega) \\ \hat{w}(j\omega) \esmtx^*
    \bar{\Pi}_2(j\omega)
    \bsmtx \hat{v}(j\omega) \\ \hat{w}(j\omega) \esmtx
    d\omega
    \ge 0
  \end{align} 
  A factorization for this multiplier is given by $M = \bsmtx 1 & 0 \\
  0 & -1 \esmtx$ and $\Psi = \bsmtx \phi_2(j \omega) & 0 \\ 0 & 1
  \esmtx$. If $w = \St v$ then $z:=\Psi \bsmtx v \\ w \esmtx = \bsmtx
  \phi_2 v \\ w \esmtx$.  Hence the time-domain form for the IQC
  is
  \begin{align}
    \int_0^\infty z(t)^T M z(t) \, dt =
    \int_0^\infty z_1^2(t) - z_2^2(t) \, dt \ge 0
  \end{align}
  where $z_2 = w$ and $z_1 = \phi_2 v$. In other words, $w$ has $L_2$
  norm less than the filtered signal $\phi_2 v$, i.e. $\|w \| \le \|
  \phi_2 v\|$.
\end{ex}
\vspace{0.1in}

It is important to note that the time domain constraint
\cref{eq:iqctd} holds, in general, only over infinite time intervals.
The term \underline{hard IQC} was introduced in \cite{Megretski1997}
referring to the following more restrictive property: $\Delta$
satisfies the IQC defined by $\Pi$ and $\int_0^T z(t)^T M z(t) \ dt
\ge 0$ holds for all $T \ge 0$, $v \in L_{2e}^{m_1}[0,\infty)$ and
$w=\Delta(v)$.  By contrast, IQCs for which the time domain constraint
need not hold over all finite time intervals are called
\underline{soft IQCs}.  Hard and soft IQCs were later generalized in
\cite{megretskiCSH} to include the effect of initial conditions and
the terms were renamed complete and conditional IQCs,
respectively. The hard/soft terminology will be used here.

This distinction between hard and soft IQCs is important because the
dissipation inequality theorems developed below require the use of
time-domain constraints that hold over all finite-time intervals.  One
issue is that the factorization of $\Pi$ as $\Psi^\sim M \Psi$ is not
unique and hence there is ambiguity surrounding the hard/soft IQC
terminology.  As a result, the characterizations of hard and soft are
not inherent to the IQC multiplier $\Pi$ but instead depend on the
factorization $(\Psi,M)$.  A more precise definition that refers to the
factorization $(\Psi,M)$ is now given.

\vspace{0.1in}
\begin{defin}
  Let $\Pi$ be factorized as $\Psi^\sim M \Psi$ with $\Psi$ stable.
  Then $(\Psi,M)$ is a \underline{hard IQC factorization} of
  $\Pi$ if for any bounded, causal operator $\Delta \in
  \text{IQC}(\Pi)$ the following time-domain inequality holds
  \begin{align}
    \label{eq:hardIQCdef}
    \int_0^T z(t)^T M z(t) \ dt \ge 0
  \end{align}
  for all $T\ge 0$, $v \in L_2^{m_1}[0,\infty)$, $w=\Delta(v)$, and
  $z=\Psi \bsmtx v \\ w \esmtx$.
\end{defin}
\vspace{0.1in}

It was recently shown that a broad class of IQC multipliers have a
hard factorization \cite{Megretski2010}.  The proof uses a new min/max
theorem to obtain a lower bound on $\int_0^T z(t)^T M z(t) \, dt$. A
similar factorization result was obtained in \cite{seiler13} using a
game-theoretic interpretation. The next theorem summarizes the main
factorization result needed in order to use IQCs within the dissipation
inequality framework.

\vspace{0.1in}
\begin{theorem}
  \label{thm:hardiqc}
  Let $\Pi=\Pi^\sim \in \RL^{(m_1+m_2) \times (m_1+m_2)}$ be
  partitioned as $\bsmtx \Pi_{11} & \Pi_{21}^\sim \\ \Pi_{21} &
  \Pi_{22} \esmtx$ where $\Pi_{11} \in \RL^{m_1 \times m_1}$ and
  $\Pi_{22} \in \RL^{m_2 \times m_2}$.  Assume $\Pi_{11}(j\omega) > 0$
  and $\Pi_{22}(j\omega) < 0$ for all $\omega \in \R \cup \{
      \infty \}$. Then $\Pi$ has a hard factorization $(\Psi,M)$.
\end{theorem}
\begin{proof}
  The sign definite conditions on $\Pi_{11}$ and $\Pi_{22}$ ensure
  that $\Pi$ has a factorization $(M,\Psi)$ where $\Psi$ is square and
  both $\Psi, \Psi^{-1}$ are stable.  This follows from
  Lemmas~\ref{lem:PIfac} and \ref{lem:PNiqc} in
  Appendix~\ref{app:IQCfac}.  Moreover, Appendix~\ref{app:IQCfac}
  provides a numerical algorithm to compute this special (J-spectral)
  factorization using state-space methods.  The conclusion that
  $(M,\Psi)$ is a hard factorization follows from Theorem 2.4 in
  \cite{Megretski2010}.
\end{proof}

\subsection{Analysis of Nonlinear Delayed Systems}
\label{sec:nlstab}

This section derives analysis conditions for the nonlinear delayed
system $F_u(G,\St)$ shown in \cref{fig:Stfeedback} using dissipation
inequalities.  For concreteness the discussion focuses on $\St$
defined with constant delays but the results also hold using IQCs
valid for $\Str$ defined with varying delays.  Assume $\St$ satisfies
the IQC defined by $\Pi$ and, in addition, $\Pi$ has a hard
factorization $(\Psi,M)$.  The feedback system can be analyzed using
the interconnection structure shown in \cref{fig:analysisic} with the
system $\Psi$ appended to the input/output channels of $\St$.  The
dynamics of the interconnection in \cref{fig:analysisic} involve an
extended system of the form
\begin{equation}
  \label{eq:nlextsys}
  \begin{split}
    \dot{x} & :=  F(x,w,d) \\
    \bmtx z \\ e \\ \emtx & = H(x,w,d) \\
  \end{split}
\end{equation}
$x:=\bsmtx x_G \\ \psi \esmtx \in \R^{n_G+n_\psi}$ is the extended
state and the functions $F$ and $H$ can be easily determined from the
dynamics of $G$ and $\Psi$.  The theorem below provides a sufficient
condition for the feedback interconnection to have an induced $L_2$
gain from $d$ to $e$ that is less than or equal to $\gamma$. The
theorem is based on a dissipation inequality that uses both the hard
IQC associated with $\St$ as well as a storage function $V$ defined on
the extended state $x$.  The system $\St$ is shown as a dashed box in
\cref{fig:analysisic} because the analysis essentially replaces the
precise relation $w=\St(v)$ with the hard IQC constraint on $z$ that
specifies the signals pairs $(v,w)$ that are consistent with the
behavior of $\St$.

\begin{figure}[h]
\centering
\scalebox{1.0}{
\begin{picture}(160,125)(-10,20)
 \thicklines
 \put(75,25){\framebox(40,40){$G$}}
 \put(158,32){$e$}
 \put(115,35){\vector(1,0){35}}  
 \put(28,32){$d$}
 \put(40,35){\vector(1,0){35}}  
 \put(80,80){\dashbox(30,30){$\St$}}
 \put(42,75){$w$}
 \put(55,55){\vector(1,0){20}}  
 \put(55,55){\line(0,1){40}}  
 \put(55,95){\line(1,0){25}}  
 \put(145,75){$v$}
 \put(135,95){\vector(-1,0){25}}  
 \put(135,55){\line(0,1){40}}  
 \put(115,55){\line(1,0){20}}  
 \put(67,95){\line(0,1){30}}
 \put(67,125){\vector(-1,0){17}}
 \put(123,95){\line(0,1){40}}
 \put(123,135){\vector(-1,0){73}}
 \put(20,115){\framebox(30,30){$\Psi$}}
 \put(5,135){$z$}
 \put(20,130){\vector(-1,0){30}}
\end{picture}
} 
\caption{Analysis Interconnection Structure}
\label{fig:analysisic}
\end{figure}

\vspace{0.1in}
\begin{theorem}
\label{thm:nlstab}
Assume $F_u(G,\St)$ is well-posed and $\St$ satisfies the hard IQC
defined by $(\Psi,M)$. Then $\|F_u(G,\St)\| \le \gamma$ if there
exists a scalar $\lambda \ge 0$ and a continuously differentiable
storage function $V:\R^{n_G+n_\psi} \rightarrow \R$ such that:
\begin{enumerate}
\item[i)] $V(0) = 0$,
\item[ii)] $V(x) \ge 0$ \ $\forall x \in \R^{n_G+n_\psi}$, 
\item[iii)] The following dissipation inequality holds for all 
  $x \in \R^{n_G+n_\psi}, w \in \R^{n_v}, d \in\R^{n_d}$ 
\begin{align}
    \label{eq:distab}
    & \lambda z^TMz + \nabla V(x) \cdot F(x,w,d) 
    \le \gamma^2 d^Td - e^Te 
  \end{align}
  where $z$ and $e$ are functions of $(x,w,d)$ as defined by $H$
  in Equation~\ref{eq:nlextsys}.
\end{enumerate}
\end{theorem}
\begin{proof}
  Let $d \in L_2^{n_d}[0,\infty)$ be any input signal. From
  well-posedness of the interconnection, the interconnection
  $F_u(G,\St)$ has a solution that satisfies the dynamics in
  Equation~\ref{eq:nlextsys}.  The dissipation inequality
  (Equation~\ref{eq:distab}) can be integrated from $t=0$ to $t=T$
  with the initial condition $x(0) = 0$ to yield:
  \begin{align}
    \lambda \int_0^T z(t)^TMz(t) \, dt + V\left(x(T)\right) 
    \le & \gamma^2 \int_0^T d(t)^T d(t) \, dt \\
    \nonumber
    & - \int_0^T e(t)^T e(t) \, dt 
  \end{align}  
  It follows from the hard IQC condition, $\lambda \ge 0$, and the
  non-negativity of the storage function $V$ that
  \begin{align}
    \label{eq:bndedgain}
    \int_0^T e(t)^T e(t) \, dt \le \gamma^2
    \int_0^T d(t)^T d(t) \, dt
  \end{align}
  Hence $\|F_u(G,\St)\| \le \gamma$. 
\end{proof}
\vspace{0.1in}

The dissipation inequality (Equation~\ref{eq:distab}) is an algebraic
constraint on variables $(x,w,d)$.  The dissipation inequality only
depends on $\St$ via the constraint on $z^TMz$. Thus the dependence of
the dissipation inequality on the delay value $\tau$ appears through
the choice of the multiplier $\Pi$.  Specifically, $\Pi$ typically
depends on the value of $\tau$, e.g. $\Pi_2$ and $\Pi_3$ defined
previously. The delay value is selected and then the multiplier $\Pi$
and its hard factorization $(\Psi,M)$ are constructed.  Thus for a
given delay value $\tau$, \cref{thm:nlstab} provides convex conditions
on $V$, $\lambda$, and $\gamma$ that are sufficient to upper bound the
$L_2$ gain of $F_u(G,\St)$.

This leads to a useful numerical procedure if additional assumptions
are made on $G$. If the dynamics of $G$ (Equation~\ref{eq:G}) are
described by polynomial vector fields then the functions $F$ and $H$
in the extended system are also polynomials.  If the storage function
$V$ is also restricted to be polynomial then the dissipation
inequality (Equation~\ref{eq:distab}) and non-negativity condition $V
\ge 0$ are simply global polynomial constraints. In this case the
search for a feasible storage function $V$ and scalars $\lambda$,
$\gamma$ can be formulated as a sum-of-squares (SOS) optimization
\cite{parrilo00,parrilo03,lasserre01}.  For fixed delay $\tau$ this
yields a convex optimization to upper bound the $L_2$ gain of
$F_u(G,\St)$.  In addition, bisection can be used to find the largest
delay $\bar{\tau}$ such that the gain from $d$ to $e$ remains
finite.  If the QC multiplier $\Pi$ covers $\St$ for all $\tau \in
[0,\bar \tau]$ then $\bar \tau$ is a lower bound on the true delay
margin. It is a lower bound because the dissipation inequality is only
a sufficient condition.  An example of this SOS method is given in
\cref{sec:nlexample}.  One issue with the SOS approach is that the
required computation grows rapidly with the degree and number of
variables contained in the polynomial constraint.  This currently
limits the proposed approach to situations where the extended system
roughly involves a cubic vector field and state dimension less than
$7-10$.

It should also be noted that the multiple IQCs can be incorporated in
the analysis.  Specifically, assume $\St$ satisfies the hard IQCs
defined by $(\Psi_k,M_k)$ for $k=1,\cdots,N$.  Each $\Psi_k$ can be
appended to the inputs/outputs of $\St$ to yield a filtered output
$z_k$. \cref{thm:nlstab} remains valid if the dissipation inequality
(Equation~\ref{eq:distab}) is modified to include the term
$\sum_{k=1}^N \lambda_k z_k^T M_k z_k$ for any constants $\lambda_k
\ge 0$.  In this case the extended system includes the dynamics of $G$
as well as the dynamics of each $\Psi_k$ ($k=1,\cdots,N$).  The
stability analysis consists of a search for the storage $V$, gain
bound $\gamma$, and the constants $\lambda_k$ that lead to feasibility
of the three conditions in Theorem~\ref{thm:nlstab}.  This approach
enables many IQCs for $\St$ to be incorporated into the analysis.  



\subsection{Analysis of Parameter-Varying Delayed Systems}
\label{sec:lpvstab}

A similar analysis condition can be derived for parameter-varying
delayed systems. In particular, Linear Parameter Varying (LPV) systems
are a class of linear systems whose state space matrices depend on a
time-varying parameter vector $\rho: \R^+ \rightarrow \R^{n_\rho}$.
The parameter is assumed to be a continuously differentiable function
of time and admissible trajectories are restricted, based on physical
considerations, to a known compact subset $\mathcal{P} \subset
\R^{n_\rho}$.  The state-space matrices of an LPV system are
continuous functions of the parameter, e.g. $A_G: \mathcal{P}
\rightarrow \R^{n_x \times n_x}$.  Define the LPV system $G_\rho$ with
inputs $(w,d)$ and outputs $(v,e)$ as:
\begin{equation}
  \label{eq:lpv}
  \begin{split}
    \dot{x}_G(t) &= A_G(\rho(t)) x_G(t) 
             + B_G(\rho(t))  \bsmtx w(t) \\ d(t) \esmtx \\
    \bsmtx v(t) \\ e(t) \esmtx & = C_G(\rho(t)) x_G(t) 
             + D_G(\rho(t))  \bsmtx w(t) \\ d(t) \esmtx \\
  \end{split}
\end{equation}
The state matrices at time $t$ depend on the parameter vector at time
$t$. Hence, LPV systems represent a special class of time-varying
systems.  Throughout this section the explicit dependence on $t$ is
occasionally suppressed to shorten the notation.

By loop-shifting, a delayed LPV system can be modeled as
$F_u(G_\rho,\St)$ where $w=\St(v)$. This is similar to the
interconnection shown in Figure~\ref{fig:Stfeedback} but with $G_\rho$
as the ``nominal'' system.  As a slight abuse of notation,
$\|F_u(G_\rho,\St)\|$ will denote the worst-case $L_2$ gain over all
allowable parameter trajectories:
\begin{equation}
  \label{eq:lpvgain}
  \| F_u(G_\rho, \St) \| = \sup_{\rho \in \mathcal{P}}  \ 
  \sup_{0\ne d \in L_2^{n_d}[0,\infty),  \ x_G(0)=0}  \frac{\|e\|}{\|d\|}
\end{equation}
As in the previous subsection, assume $\St$ satisfies the IQC defined
by $\Pi$ and that $\Pi$ has a hard factorization $(\Psi,M)$. The
stability and performance of $F_u(G_\rho,\St)$ can be assessed by
appending the system $\Psi$ to the input/output channels of $\St$ (as
in \cref{fig:analysisic}).  The dynamics of the interconnection in
\cref{fig:analysisic} depend on an extended LPV system of the
form:
\begin{equation}
\label{eq:lpvic}
\begin{split}
 \dot{x} &= A(\rho)  x  + B_1(\rho) w + B_2(\rho) d \\
z &=  C_1(\rho)  x  + D_{11}(\rho)w + D_{12}(\rho)d \\
e &=  C_2(\rho)  x  + D_{21}(\rho)w + D_{22}(\rho)d, \\
\end{split}
\end{equation}
where the state vector is $x:=\bsmtx x_G \\ \psi \esmtx \in
\R^{n_G+n_\psi}$ with $x_G$ and $\psi$ denoting the state vectors of the
LPV system $G_\rho$ (Equation~\ref{eq:lpv}) and the filter $\Psi$
(Equation~\ref{eq:ltipsi}), respectively.  A dissipation inequality
can be formulated to upper bound the worst-case $L_2$ gain of
$F_u(G_\rho,\St)$ using the system \cref{eq:lpvic} and the time domain
IQC \cref{eq:hardIQCdef}. This dissipation inequality is concretely
expressed as a linear matrix inequality in the following theorem. The
theorem is stated assuming a single multiplier $\Pi$ for $\St$ but
many IQC multipliers can be included as described in the previous
section.

\vspace{0.1in}
\begin{theorem}
\label{thm:lpvstab}
Assume $F_u(G_\rho,\St)$ is well posed and $\St$ satisfies the hard
IQC defined by $(\Psi,M)$.  Then $\|F_u(G_\rho,\St)\| \le \gamma$ if
there exists a scalar $\lambda \ge 0$ and a matrix $P=P^T \in
\R^{n_x+n_\psi}$ such that $P\ge 0$ and for all $\rho \in \mathcal{P}$:
\begin{equation}
  \label{eq:brllmi}
  \begin{split}
    \bmtx A^T P + PA & P B_1 & P B_2 \\ B_1^TP & 0 & 0 
    \\ B_2^TP & 0 & -\gamma^2 I \emtx 
    + \bmtx C_2^T \\ D_{21}^T \\ D_{22}^T \emtx 
    \bmtx C_2 & D_{21} & D_{22} \emtx \\
    + \lambda \bmtx C_1^T \\ D_{11}^T \\ D_{12}^T \emtx 
    M \bmtx C_1 & D_{11} & D_{12} \emtx < 0 \\
  \end{split}
\end{equation}
\end{theorem}
In \cref{eq:brllmi} the dependency of the state space matrices on $\rho$
has been omitted to shorten the notation.
\begin{proof}
  Define a storage function $V:R^{n_x+n_\psi} \rightarrow
  \mathbb{R}^+$ by $V(x) = x^TPx$.  Left and right multiply
  \cref{eq:brllmi} by $[x^T, w^T,d^T]$ and $[x^T, w^T, d^T]^T$ to show
  that $V$ satisfies the dissipation inequality:
  \begin{align}
    \label{eq:distab2}
    & \lambda z(t)^TMz(t) + \dot{V}(t)  \le \gamma^2 d(t)^Td(t) - e(t)^Te(t) 
  \end{align}
  The remainder of the proof is similar to that given for \cref{thm:nlstab}.
\end{proof}
\vspace{0.1in}


The analysis of the delayed LPV system reduces to a set of parameter
dependent LMIs.  These are infinite dimensional and hence they are
typically approximated by finite-dimensional LMIs evaluated on a grid
of parameter values.  In this case the search for the matrix $P$ and
scalars $\lambda$, $\gamma$ can be performed as a semidefinite
programming optimization \cite{befb94}.  If the LPV system has a
rational dependence on the parameters then a finite dimensional LMI
condition can be derived (with no approximation) using the techniques
in \cite{Packard1994,Apkarian1995}. In addition, there may be known
bounds on the parameter rates of variation.  \cref{thm:lpvstab} does
not incorporate such knowledge and hence this is called a
rate-unbounded analysis condition. \cref{thm:lpvstab} can be easily
extended to include rate-bounds using parameter-dependent storage
functions as described in \cite{Wu1996}.  Finally, if the system
$G_\rho$ is LTI then Equation~\ref{eq:distab2} reduces to a single LMI
condition.  The nonlinear dissipation inequality in \cref{thm:nlstab}
is equivalent to exactly the same LMI condition when $G$ is LTI
and $V$ is a quadratic function of $x$. In other words,
Theorems~\ref{thm:nlstab} and \ref{thm:lpvstab} are the same for
LTI dynamics and quadratic storage functions.

%
%

\section{Numerical Examples}
\label{sec:examples}

This section presents numerical examples to assess the stability and
performance for nonlinear and LPV delayed systems.

\subsection{Nonlinear Delayed System}
\label{sec:nlexample}

Consider the classical feedback loop shown in \cref{fig:cloop} where
$\Dt$ is a constant delay.  $L$ is a nonlinear system described by the
following ODE:
\begin{equation}
\label{eq:Lex}
\begin{split}
\dot{x}_G & = \bmtx -49 & 0 \\ 1 & 0 \emtx x_G  
    + \bmtx 8 \\ 0 \emtx \tilde{w} 
    + p(x_G) \\
y &= \bmtx -4.5 & 1.5 \emtx x_G
\end{split}
\end{equation}
where $p(x_G):= \bmtx 2 x_{G,1}^2 + 3 x_{G,2}^2 - 0.2 x_{G,1}^3, \
-x_{G,2}^3 \emtx^T$.  As described in \cref{sec:probform}, a
loop-shift can be performed to bring the classical feedback loop into
the form $F_u(G,\St)$ shown in \cref{fig:Stfeedback}.

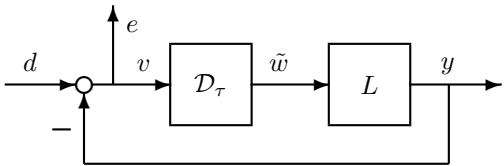
\begin{figure}[htbp]
\centering
\scalebox{1.0}{
\begin{picture}(203,60)(20,20)
 \thicklines 
 \put(20,50){\vector(1,0){27}}
 \put(27,55){$d$}
 \put(38,33){\line(1,0){7}}
 \put(50,50){\circle{6}}
 \put(53,50){\vector(1,0){30}}
 \put(70,55){$v$}
 \put(61,50){\vector(0,1){30}}
 \put(66,70){$e$}
 \put(83,35){\framebox(30,30){$\Dt$}}
 \put(113,50){\vector(1,0){30}}
 \put(120,55){$\tilde w$}
 \put(143,35){\framebox(30,30){$L$}}
 \put(173,50){\vector(1,0){35}}
 \put(188,50){\line(0,-1){30}}
 \put(188,20){\line(-1,0){138}}
 \put(50,20){\vector(0,1){27}} 
 \put(185,56){$y$}
\end{picture}
} 
\caption{Classical Feedback Loop}
\label{fig:cloop}
\end{figure}


First, a linear analysis is performed to gain insight into the use of
QCs for analysis.  Let $G_{lin}$ denote the linearization of $G$
around $x_G=0$.  This linearization is obtained by neglecting the
higher order terms in \cref{eq:Lex}, i.e. neglecting $p(x_G)$.  An
estimate for the delay margin of $F_u(G_{lin},\St)$ can be computed
from the LMI condition in \cref{thm:lpvstab}. Bisection is used to
find the largest time delay for which the gain from $d$ to $e$ is
finite. A number of solvers exist for this class of systems and here
the Matlab \texttt{LMILab} toolbox was used.  Using the standard
multipliers $\Pi_1$ and $\bar{\Pi}_2$ gives a delay margin of 0.06sec.
However, using the multipliers $\Pi_1$ and the new ``small'' circle
multiplier $\bar{\Pi}_3$ yields a significantly larger delay margin
estimate of 1.96sec.  The exact delay margin for $F_u(G_{lin},\St)$
can be estimated from the frequency response of $G_{lin}$.  The
frequency response gives a delay margin of 2.05 sec with a critical
frequency 0.361 rad/sec.  These results are summarized in
\cref{tab:nldm}.

\begin{table}[h]
  \centering
  \begin{tabular}{llr} 
    \hline
    System & Method & Delay Margin \\ 
    \hline
    Linear, $G_{lin}$  & 
          IQC with $\Pi_1$ and $\bar \Pi_2$ & 0.06sec \\
    Linear, $G_{lin}$  & 
          IQC with $\Pi_1$ and $\bar \Pi_3$ & 1.96sec \\
    Linear, $G_{lin}$ & Frequency Response & 2.05 sec \\ 
    \hline \hline
    Nonlinear, $G$   
         &  IQC with $\Pi_1$ and $\bar \Pi_3$ & 1.09sec \\
    \hline
  \end{tabular}
  \vspace{0.1in}
  \caption{Summary of Delay Margins} \label{tab:nldm}
\end{table}

The results can be interpreted in the complex plane as shown in
\cref{fig:NyquistQCs}.  The figure shows the circle constraints
described by $\Pi_1$ and $\bar{\Pi}_3$ at the critical frequency 0.361
rad/sec (dotted blue). The figure also shows the circle described by
the optimal multiplier $\Pi_{opt} := \lambda_1 \Pi_1 + \lambda_3
\bar{\Pi}_3$ (solid blue). The coefficients in the optimal conic
combination are given by $\lambda_1 = 5310$ and $\lambda_3 = 7210$.
Finally, the figure shows the location of $\hatSt(j\omega) =
e^{-j\omega \tau} -1$ evaluated at the critical frequency and delay
margin estimate of $\tau=1.96$ sec (small red circle).  The location
of $\hatSt$ lies inside the circle described by $\Pi_{opt}$ at the
critical frequency. In fact, $\hatSt$ lies inside the circle described
by $\Pi_{opt}$ at all frequencies.  This is consistent with the
following rough frequency-domain interpretation of the dissipation
inequality result: $F_u(G,\Delta)$ is stable for LTI perturbations
$\Delta$ whose Nyquist plot lies within the circles described
$\Pi_{opt}$.

\begin{figure}[ht!]
  \centering
    {\input{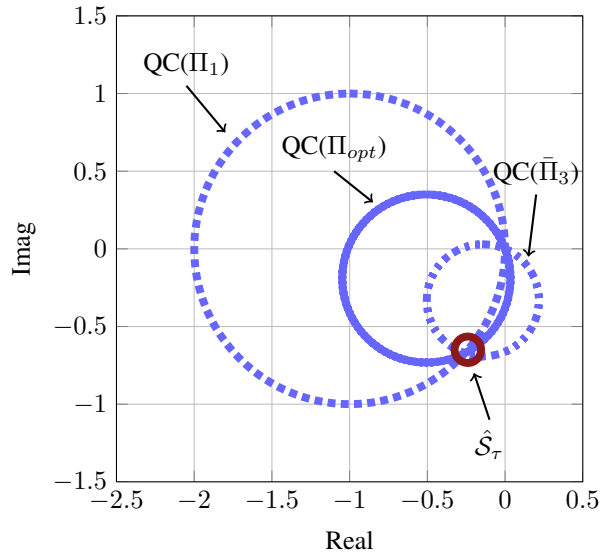}}
  \caption{Nyquist plot showing quadratic constraints.}
  \label{fig:NyquistQCs}
\end{figure}

This frequency domain interpretation can also be used to explain the
poor delay margin bound obtained with $\Pi_1$ and $\bar{\Pi}_2$.
Briefly, the transfer function from $w$ to $v$ (upper left corner of $G_{lin}$) is given by $-T=\frac{-L_{lin}(s)}{1+L_{lin}(s)}$, i.e.  the negative
complementary sensitivity function. Standard robust control techniques
\cite{Skogestad2005} can be used to show that
$\Delta:=\frac{-1}{T(j\omega_0)}$ is a complex perturbation that
causes $F_u(G_{lin},\Delta)$ to have poles at $\pm j\omega_0$. At high frequencies these ``small'' destabilizing
perturbations follow a vertical line in the complex plane with real
part equal to 0.37. This constrains the use of
$\bar{\Pi}_2$ because this multiplier defines a circle centered at the
origin with radius greater than $\hatSt$.  The destabilizing
perturbations lie within this circle unless the time delay is chosen
to be small. This issue can be eliminated by frequency weighting the
multiplier $\bar{\Pi}_2$, e.g. weighting each entry by $\left|
  \frac{s+1}{10s+1} \right|^2$.  However, this increases the state
order of the multiplier and hence the computation time.

Next consider the analysis of the delayed nonlinear system.  An
estimate for the delay margin of $F_u(G,\St)$ can be computed from the
dissipation inequality in \cref{thm:nlstab}. 
A number of solvers exist for the corresponding sum-of-squares optimizations.
Here the SOSOPT toolbox \cite{sosopt13} was used for all computations.
Using the multipliers $\Pi_1$ and $\bar{\Pi}_3$ yields a delay margin
of $\bar{\tau}_{NL}=1.09$ sec for the nonlinear system. Thus the
nonlinearities significantly degrade the delay margin bound.  The
bisection was initialized with the bounds $\bar{\tau}_{NL} \in
[0,5]$. The computation took 5.2 secs to perform 13 bisection steps
to achieve a tolerance of 0.001.

\begin{figure}[ht!]
  \centering
  {
%
%
%
\definecolor{mycolor1}{rgb}{0.4,0.4,1}%
\definecolor{mycolor2}{rgb}{0.546875,0.09765625,0.09765625}%
\begin{tikzpicture}

\begin{axis}[%
width=.8\columnwidth,
height=0.6\columnwidth,
scale only axis,
xmin=0,
xmax=2.04514019564538,
xlabel={Time Delay [s]},
ymin=0,
ymax=38.2942995679458,
ylabel={L2 gain[-]}
]
\addplot [
color=lightgray!80!black,
solid,
line width=3.0pt,
forget plot
]
table[row sep=crcr]{
0.019586181640625 3.82942995679458\\
0.120609644839638 4.04756368511881\\
0.221633108038651 4.27862623117196\\
0.322656571237664 4.53206101271388\\
0.423680034436678 4.81922399067115\\
0.524703497635691 5.1456662281756\\
0.625726960834704 5.50512429255725\\
0.726750424033717 5.95734509887653\\
0.82777388723273 6.41480130224412\\
0.928797350431743 7.03944846621021\\
1.02982081363076 7.76224588981843\\
1.13084427682977 8.63159006226498\\
1.23186774002878 9.69916365376377\\
1.3328912032278 10.9864386323379\\
1.43391466642681 12.3752753524223\\
1.53493812962582 15.1325533528799\\
1.63596159282484 19.415576276356\\
1.73698505602385 23.1942706458788\\
1.83800851922286 31.9590804796179\\
1.93903198242187 68.5352856747063\\
};
\addplot [
color=mycolor1,
dashed,
line width=3.0pt,
forget plot
]
table[row sep=crcr]{
0.019586181640625 4.02279181162868\\
0.120609644839638 4.19250315121006\\
0.221633108038651 4.4604429004154\\
0.322656571237664 4.77509232934449\\
0.423680034436678 5.08665510401283\\
0.524703497635691 5.40797997564419\\
0.625726960834704 5.77733635578066\\
0.726750424033717 6.19711605942828\\
0.82777388723273 6.72491114019019\\
0.928797350431743 7.38975818222635\\
1.02982081363076 8.21073331553566\\
1.13084427682977 9.23654208190632\\
1.23186774002878 10.5372865651836\\
1.3328912032278 12.2539507408787\\
1.43391466642681 14.6426367664813\\
1.53493812962582 18.1707638010397\\
1.63596159282484 23.8914220899939\\
1.73698505602385 34.8015447296297\\
1.83800851922286 63.9551544936924\\
1.93903198242187 385.363166793634\\
};
\addplot [
color=mycolor2,
dash pattern=on 1pt off 3pt on 3pt off 3pt,
line width=3.0pt,
forget plot
]
table[row sep=crcr]{
0.01090087890625 5.22472701314652\\
0.06712646484375 5.45902901345451\\
0.12335205078125 6.01519805928224\\
0.17957763671875 7.05720128492196\\
0.23580322265625 8.41198711363584\\
0.29202880859375 9.89330083627091\\
0.34825439453125 11.4758486473164\\
0.40447998046875 13.1871157740697\\
0.46070556640625 15.0778421498112\\
0.51693115234375 17.2200314657351\\
0.57315673828125 19.7147813056699\\
0.62938232421875 22.7091835228467\\
0.68560791015625 26.4283806499103\\
0.74183349609375 31.2389131524583\\
0.79805908203125 37.7849486244559\\
0.85428466796875 47.3183615967144\\
0.91051025390625 62.6444421196942\\
0.96673583984375 91.6282017250966\\
1.02296142578125 168.004201545399\\
1.07918701171875 950.774552425732\\
};

\draw[thick, ->] (40, 32) node[above] {NonLin-IQC} -- (65, 29) ;

\draw[thick, ->] (140, 25) node[left] {Lin-IQC} -- (160, 25) ;

\draw[thick, ->] (150, 7) node[below] {Lin-Freq. Resp.} -- (145, 11) ;

\end{axis}

\end{tikzpicture}%
}
  \caption{Induced $L_2$ gain versus delay}
  \label{fig:NLgain}
\end{figure}
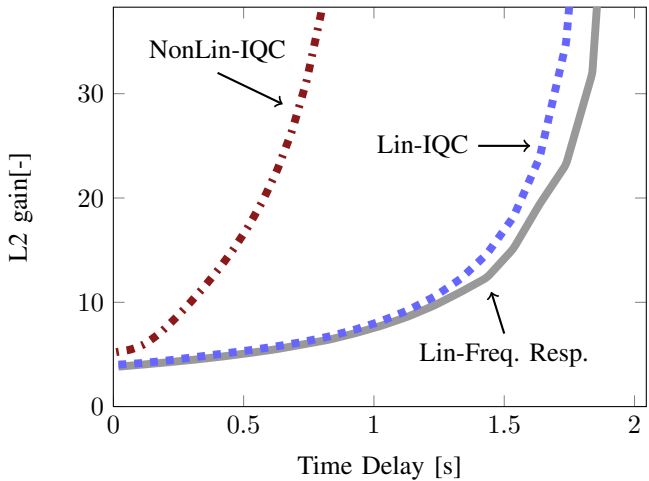

The $L_2$ gain of the delayed nonlinear system from reference $d$ to
tracking error $e$ can also be computed from the dissipation
inequality in \cref{thm:nlstab}.  Figure~\ref{fig:NLgain} shows the
gain of the nonlinear system for delays below the delay margin
estimate, $\tau < \bar{\tau}_{NL}=1.09$ sec (green dash-dot).  The
multipliers $\Pi_1$ and $\bar{\Pi}_3$ were used to compute this curve.
It took 7.5 sec to evaluate the gain on a grid of 20 delay
values. For comparison the figure also shows the gain of the linear
system $F_u(G_{lin},\St)$ computed using two methods. The red dashed
curve is the gain computed using the LMI condition in
\cref{thm:lpvstab} also with multipliers $\Pi_1$ and $\bar{\Pi}_3$.
The blue solid curve is the true induced $L_2$ gain of the linear
system estimated from the frequency response of
$F_u(G_{lin},\St)$. The two linear results are close which provides
confidence in the upper bounds computed via the IQC and dissipation
inequality condition.  The figure also shows that the nonlinear
delayed system has significantly larger gain as compared to the
linearized system.  This again indicates that the nonlinearities
significantly degrade the performance.

\subsection{LPV Delayed System}

The second example is an LPV time-delayed system representing a cutter
used in milling. The system $G_\rho$, taken from \cite{Zhang2002}, can be
written as
\begin{equation}
\begin{split}
\dot{x}_G = &\bmtx 0 & 0 & 1 & 0 \\ 0 & 0 & 0 & 1 \\ -(10+0.171 k) + 0.5 k\rho & 10 & 0 & 0 \\ 5 & -15 & 0 & -0.25 \emtx x_G  \\
 & + \bmtx 0 & 0 & 0 & 0 \\ 0 & 0 & 0 & 0\\ 0.171 k - 0.5 k\rho & 0 & 0 & 0\\ 0 & 0 & 0 & 0 \emtx \tilde w \\
\tilde w = &\mathcal{D}_\tau(x_G)
\end{split}
\end{equation}
where $k$ is the cutting stiffness and $\rho \in [-1, 1]$ is an artificial scheduling parameter depending on the angular position of the blade. Since $G_\rho$ only depends affinely on the parameter $\rho$, it is sufficient to only consider the vertices, i.e. $\rho = -1$ and $\rho = 1$, for the following analyses. The goal of the benchmark is to find the time delay margin of the plant for a given cutting stiffness $k$. As in the previous example a loop shift is performed to bring the system into the form
$F_u(G_\rho, \St)$.

The results of the analysis are shown in \cref{fig:ex3results}. The
IQC approach is based on the LMI condition \cref{thm:lpvstab} and
solved via bisection using the Matlab \texttt{LMILab} toolbox. The IQC
analysis is performed using the combination of multipliers $\Pi_1$ and
$\bar{\Pi}_2$.  The IQC analysis is also performed using $\Pi_1$ and
$\bar{\Pi}_3$. Finally the method from \cite{Zhang2002} based on
Lyapunov-Krasovskii functionals is shown for comparison with the IQC
approach.  It is shown in \cite{Zhang2002} that the system is stable
independent of the delay for stiffness values $k \leq 0.267$. As can
be seen in the figure, all three methods capture this behavior
well. For higher stiffness values, the IQC method with using
multipliers $\Pi_1$ and $\bar{\Pi}_3$ yields improved margin bounds as
compared to the analysis condition from \cite{Zhang2002}.  It should
be noted that the improved Lyapunov-Krasovskii condition in
\cite{Gu1997} achieves results similar to the best IQC results
obtained with $\Pi_1$ and $\bar{\Pi}_3$.


\begin{figure}[htbp]
\centering 
	{
%
%
%
\definecolor{mycolor1}{rgb}{0.546875,0.09765625,0.09765625}%
\definecolor{mycolor2}{rgb}{0.4,0.4,1}%
\begin{tikzpicture}

\begin{axis}[%
width=.8\columnwidth,
height=0.6\columnwidth,
scale only axis,
xmin=0.15,
xmax=0.5,
xlabel={cutting stiffness [-]},
ymin=0,
ymax=2,
ylabel={time delay margin [s]},
legend style={draw=black,fill=white,legend cell align=left}
]

\addplot [
color=lightgray!80!black,
solid,
line width=3.0pt,
]
table[row sep=crcr]{
0.26 99.9992370605469\\
0.265 0.84991455078125\\
0.27 0.7965087890625\\
0.28 0.73699951171875\\
0.3 0.6622314453125\\
0.35 0.547027587890625\\
0.4 0.470733642578125\\
0.45 0.414276123046875\\
};
\addlegendentry{Cor. 5.3. \cite{Zhang2002} };

\addplot [
color=mycolor1,
dash pattern=on 1pt off 3pt on 3pt off 3pt,
line width=3.0pt,
]
table[row sep=crcr]{
0.26 99.9992370605469\\
0.265 1.49154663085938\\
0.27 1.397705078125\\
0.28 1.31683349609375\\
0.3 1.2054443359375\\
0.35 1.0009765625\\
0.4 0.860595703125\\
0.45 0.756072998046875\\
};
\addlegendentry{$\Pi_1 \text{ and } \bar{\Pi}_3$};

\addplot [
color=mycolor2,
dotted,
line width=3.0pt,
]
table[row sep=crcr]{
0.26 99.9992370605469\\
0.265 1.1138916015625\\
0.27 0.9796142578125\\
0.28 0.872802734375\\
0.3 0.753021240234375\\
0.35 0.591278076171875\\
0.4 0.496673583984375\\
0.45 0.4302978515625\\
};
\addlegendentry{$\Pi_1 \text{ and } \bar{\Pi}_2$};

\end{axis}
\end{tikzpicture}%
}
\caption{Time Delay Margin vs Cutting Stiffness}
\label{fig:ex3results}
\end{figure}
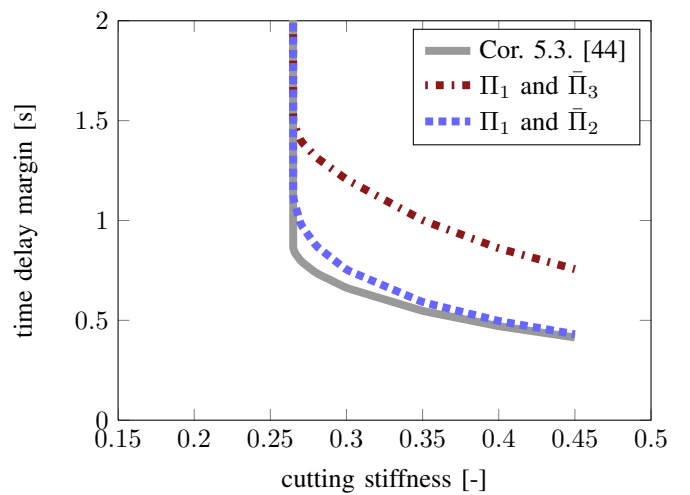

\section{Conclusions}

This paper developed stability and performance analysis conditions for
nonlinear and LPV time-delayed systems.  The approach bounds the
behavior of the time delay using integral quadratic constraints
(IQCs).  IQCs are typically specified as frequency domain constraints
on the inputs and outputs of an operator. For constant time delays,
such IQCs were shown to have an intuitive geometric interpretation at
each frequency.  This intuition led to the construction of new IQCs
for both constant and varying delays.  Dissipation inequalities were
then provided that incorporate IQCs into the analysis of delayed
systems.  This analysis approach applied existing results to obtain an
equivalent time-domain intepretation for IQCs.  Finally, numerical
examples were provided to demonstrate the proposed methods. While this
paper only considers a single time delay, the framework can easily
incorporate additional delays and uncertainties analysis. Future work
will explore the connections between the IQCs/dissipation
inequality approach and standard Lyapunov methods.



\section*{Acknowledgments}

This work was partially supported by NASA under Grant No. NRA
NNX12AM55A entitled ``Analytical Validation Tools for Safety Critical
Systems Under Loss-of-Control Conditions'', Dr. C. Belcastro technical
monitor.  This work was also partially supported by the National
Science Foundation under Grant No. NSF-CMMI-1254129 entitled ``CAREER:
Probabilistic Tools for High Reliability Monitoring and Control of
Wind Farms''. Any opinions, findings, and conclusions or
recommendations expressed in this material are those of the authors
and do not necessarily reflect the views of NASA or NSF.


\bibliographystyle{plain}
\bibliography{IQCandSOS}


\appendix

\subsection{Time-varying IQC}
\label{app:newiqc}

This appendix provides a proof that if $r<1$ then $\Str \in
\text{IQC}(\Pi_6)$ where the multiplier $\Pi_6$ is defined in
\cref{sec:vtdqc}.  Two preliminary lemmas are required to prove this
result.

\vspace{0.1in}
\begin{lemma}
  \label{lem:hfbound}
  If $r<1$ then $\| \Str - \frac{1}{2} \Stbar \| \le 1 +
  \frac{1}{\sqrt{1-r}}$.
\end{lemma}
\begin{proof}
  Apply the triangle inequality as well as the definitions of 
  $\Str$ and $\Stbar$ to obtain the following bound:
  \begin{align}
    \| \Str - \frac{1}{2} \Stbar \| \le  
    \| \Dtr \| + \frac{1}{2} \| \Dtbar \| + \frac{1}{2}
  \end{align}
  The varying delay is bounded as $\| \Dtr \| \le
  \frac{1}{\sqrt{1-r}}$ \cite{Gu2002,Kao2007} while the constant delay
  is bounded by $\| \Dtbar \| \le 1$.
\end{proof}

\vspace{0.1in}
\begin{lemma}
  \label{lem:lfbound}
$\| ( \Str - \frac{1}{2} \Stbar ) \circ \frac{1}{s} \|
\le \frac{1}{2} \bar{\tau} $.
\end{lemma}
\begin{proof}
  The proof is only sketched as it is similar to that
  given for Lemma 1 in \cite{Kao2004}.  To simplify notation 
  define $\Delta:=( \Str - \frac{1}{2} \Stbar ) \circ \frac{1}{s}$.
  Consider $w=\Delta v$ for some $v\in L_2[0,\infty)$ and
  define $y(t):=\int_0^t v(\alpha) d\alpha$.  Thus
  $w=( \Str - \frac{1}{2} \Stbar )(y)$ which, after some algebra,
  gives
  \begin{align}
    w(t) = \int_{t-\bar \tau}^t s(\alpha) v(\alpha) d\alpha
  \end{align}
  where $s(\alpha) = +\frac{1}{2}$ for $\alpha \in [t-\bar \tau,t-\tau(t)]$
  and $s(\alpha) = -\frac{1}{2}$ for $\alpha \in [t-\tau(t),t]$.
  The Cauchy-Schwartz inequality can then be used to show
  \begin{align}
    w(t)^2 \le \frac{\bar \tau}{4} \int_{t-\bar{t}}^t v^2(\alpha) d\alpha
  \end{align}
  Integrate this inequality from $t=0$ to $t=\infty$ and perform
  a change of variables to obtain
  $\|w\|^2 \le \frac{\bar{\tau}^2}{4} \| v\|^2$. 
\end{proof}

\vspace{0.1in}
\begin{theorem}
\label{thm:newiqc}
Let $\phi_6(s)$ be the transfer function defined in
Equation~\ref{eq:phi6}.  If $r<1$ then $\| ( \Str - \frac{1}{2} \Stbar
) \circ \phi_6^{-1} \| \le 1$.
\end{theorem}
\begin{proof}
  The proof is only sketched as it is essentially the same as that
  given for Proposition 2 in \cite{Kao2007}. Let $v\in L_2$ be an
  input signal and $\hat{v} := \FT(v)$ its corresponding Fourier
  Transform. Decompose $v$ as $v_L+v_H$ where $v_L$ and $v_H$ are the
  low and high frequency components, respectively.  Specifically, the
  low-frequency content is defined in the frequency domain by
  $\hat{v}_L(j\omega):=\hat{v}(j\omega)$ if $ |\omega| \le
  \frac{2}{\bar \tau} \left( 1 + \frac{1}{\sqrt{1-r}} \right)$ and
  $\hat{v}_L(j\omega):=0$ otherwise. The high-frequency content is
  defined similarly.  To simplify notation, define $\Delta = ( \Str -
  \frac{1}{2} \Stbar ) \circ \phi_6^{-1}$.  Then using the linearity of
  $\Delta$ and the triangle inequality yields
  \begin{align}
    \| \Delta v \| \le \| \Delta v_L \| + \| \Delta v_H \|
  \end{align}
  Lemmas~\ref{lem:hfbound} and \ref{lem:lfbound} bound the gains on
  the high and low frequency components by $\| \Delta v_H \| \le \|
  v_H \|$ and $\| \Delta v_L \| \le \| v_L \|$.  Thus $\| \Delta v \|
  \le \| v_L \| + \| v_H \| = \|v \|$.
\end{proof}
\vspace{0.1in}

The bound in Theorem~\ref{thm:newiqc} can be equivalently expressed as
a quadratic, frequency-domain constraint on the input/output signals
of $\Str$. This gives the desired result that $\Str$ satisfies the IQC
defined by the multiplier $\Pi_6$.

\subsection{IQC Factorizations}
\label{app:IQCfac}

This appendix provides specific numerical procedures to factorize
$\Pi = \Pi^\sim \in \RL^{m\times m}$ as $\Psi^\sim M \Psi$. Such
factorizations are not unique and this appendix presents
two useful factorizations.


First, let $(A_\pi,B_\pi,C_\pi,D_\pi)$ be a minimal state-space
realization for $\Pi$.  Separate $\Pi$ into its stable and unstable
parts $\Pi = G_S + G_U$.  Let $(A,B,C,D_\pi)$ denote a state space
realization for the stable part $G_S$. The matrix $A$ is Hurwitz since
$G_S$ is stable. The assumptions on $\Pi$ can be used to show that the
poles of $\Pi$ are symmetric about the imaginary axis and, moreover,
$G_U$ has a state space realization of the form $(-A^T,-C^T,B^T,0)$
(Section 7.3 of \cite{francis87}). Thus $\Pi = G_S + G_U$ can be
written in the form $\Pi = \Psi^\sim M \Psi$ where
\begin{align}
\Psi(s) & := \bmtx (sI-A)^{-1} B \\ I \emtx \\
M & := \bmtx 0 & C^T \\ C & D_\pi \emtx 
\end{align}
This provides a factorization $\Pi = \Psi^\sim M \Psi$ where $M=M^T
\in \R^{n_z \times n_z}$ and $\Psi \in \RH^{n_z \times m}$.  For this
factorization $\Psi$ is, in general, non-square ($n_z\ne m$) and
it may have right-half plane zeros.

The stability theorems in this paper require a special factorization
such that $\Psi$ is square ($n_z = m$), stable, and minimum phase.
More precisely, given non-negative integers $p$ and $q$, let $J_{p,q}$
denote the signature matrix $\bsmtx I_p & 0 \\ 0 & -I_q \esmtx$.
$\Psi$ is called a $J_{p,q}$-spectral factor of $\Pi$ if $\Pi =
\Psi^\sim J_{p,q} \Psi$ and $\Psi, \Psi^{-1} \in \RH^{m\times m}$.
The term $J$-spectral factor will be used if the values of $p$ and $q$
are not important.  $J$-spectral factorizations have been used to
construct (sub-optimal) solutions to the $H_\infty$ optimal control
problem \cite{green90,kimura92,francis87}.  The next lemma provides
a necessary and sufficient condition for constructing a $J$-spectral
factorization of $\Pi$.

\vspace{0.1in}
\begin{lemma}
  \label{lem:PIfac}
  Let $\Pi \in \RL^{m\times m}$ be a multiplier in the form:
  \begin{align}
    \Pi(s) = 
    \bmtx (sI-A)^{-1} B \\ I \emtx^\sim
    \bmtx 0 & C^T \\ C & D_\pi \emtx
    \bmtx (sI-A)^{-1} B \\ I \emtx
  \end{align}
  where $A$ is Hurwitz. The following statements are equivalent:
  \begin{enumerate}
  \item $D_\pi=D_\pi^T$ is nonsingular and there exists a unique real solution
    $X=X^T$ to the following ARE
    \begin{align}
      \label{eq:are}
      A^TX + XA - (XB+C^T) D_\pi^{-1} (B^TX+C) = 0 
    \end{align}
    such that $A-BD_\pi^{-1} \left(B^TX + C \right)$ is Hurwitz.
  \item $\Pi$ has a $J_{p,q}$ spectral factorization where $p$ and $q$
    are the number of positive and negative eigenvalues of $D_\pi$,
    respectively.  Moreover, $\Psi$ is a $J_{p,q}$-spectral factor of $\Pi$
    if and only if it has a state-space realization $\left( A, B
      ,J_{p,q} W^{-T} \left( B^TX+C \right), W \right)$ where $W$ is a
    solution of $D_\pi = W^T J_{p,q} W$.
  \end{enumerate}
\end{lemma}
\begin{proof}
  This lemma is based on the canonical factorization results in
  \cite{bart79} and summarized in Chapter 7 of \cite{francis87}.  The
  precise wording of this lemma is a special case of Theorem 2.4 in
  \cite{meinsma95}.
\end{proof}
\vspace{0.1in}

The numerical procedure to construct a $J$-spectral factorization of
$\Pi=\Pi^\sim$ can be summarized by the following steps.  First,
express $\Pi$ with a minimal realization
$(A_\pi,B_\pi,C_\pi,D_\pi)$. Second, compute a state space realization
$(A,B,C,D_\pi)$ for the stable part of $\Pi$. This step can be done
using the Matlab command $\texttt{stabsep}$. Third, attempt to solve
for the stabilizing solution $X=X^T$ of the ARE in
Equation~\ref{eq:are}.  This step can be done using the Matlab command
$\texttt{care}$.  The existence of a stabilizing solution is related
to the eigenvalues/eigenspaces of a related Hamiltonian matrix (see
\cite{zhou96}).  If this step is unsuccessful, i.e. no such solution
exists, then $\Pi$ does not have a $J$-spectral factorization by
\cref{lem:PIfac}.  However, if the ARE has a unique stabilizing
solution then construct the state-space realization for $\Psi$ as
defined in Statement 2) of \cref{lem:PIfac}.  This requires the matrix
decomposition $D_\pi = W^T J_{p,q} W$ which can be computed from an
eigenvalue decomposition of $D_\pi$. This entire procedure, if
successful, yields the factorization $\Pi = \Psi^\sim M \Psi$ where
$M:=J_{p,q}$ and $\Psi, \Psi^{-1} \in \RH^{m\times m}$.

The last result of this appendix provides a simple frequency domain
condition that is sufficient for the existence of a $J$-spectral
factor of a multiplier $\Pi=\Pi^\sim$.

\vspace{0.1in}
\begin{lemma}
  \label{lem:PNiqc}
  Let $\Pi=\Pi^\sim \in \RL^{(m_1+m_2) \times (m_1+m_2)}$ be
  partitioned as $\bsmtx \Pi_{11} & \Pi_{12} \\ \Pi_{12}^\sim &
  \Pi_{22} \esmtx$ where $\Pi_{11} \in \RL^{m_1 \times m_1}$ and
  $\Pi_{22} \in \RL^{m_2 \times m_2}$.  Assume $\Pi_{11}(j\omega) > 0$
  and $\Pi_{22}(j\omega) < 0$ for all $\omega \in \R \cup \{
      \infty \}$. Then $\Pi$ has a $J_{m_1,m_2}$-spectral factorization.
\end{lemma}
\begin{proof}
  The sign definite conditions on $\Pi_{11}$ and $\Pi_{22}$ can be
  used to show that $\Pi$ has no equalizing vectors (as defined in
  \cite{meinsma95}) and hence the corresponding ARE has a unique
  stabilizing solution (Theorem 2.4 in \cite{meinsma95}).
  Details are given in \cite{seiler13}.
\end{proof}

\end{document}